\documentclass[%
 reprint,
showpacs,preprintnumbers,
 amsmath,amssymb,
 aps,
pre,
]{revtex4-1}

\usepackage{graphicx}
\usepackage{dcolumn}
\usepackage{bm}
\usepackage{physics}
\usepackage{amsthm} 
\usepackage{wasysym}

\newtheorem{theorem}{Theorem}[section]

\newcommand*\diff{\mathop{}\!\mathrm{d}}

\usepackage{todonotes}

\begin{document}

\title[Phase balancing in the Kuramoto model]{Order parameter allows classification of planar graphs based on balanced fixed points in the Kuramoto model}

\author{Franz Kaiser}
\author{Karen Alim}
 \email{karen.alim@ds.mpg.de}
\affiliation{ 
Biological Physics and Morphogenesis Group, Max Planck Institute for Dynamics and Self-Organization, 37077 G\"ottingen, Germany
}
\affiliation{Institute for Nonlinear Dynamics, Faculty of Physics, University of G\"ottingen, 37077 G\"ottingen, Germany}

\date{\today}

\begin{abstract}
Phase {\it balanced states} are a highly under-explored class of solutions of the Kuramoto model and other coupled oscillator models on networks. So far, coupled oscillator research focused on phase synchronized solutions. Yet, global constraints on oscillators may forbid synchronized state, rendering phase balanced states as the relevant stable state. If for example oscillators are driving the contractions of a fluid filled volume, conservation of fluid volume constraints oscillators to balanced states as characterized by a vanishing Kuramoto order parameter.
It has previously been shown that stable, balanced patterns in the Kuramoto model exist on circulant graphs. However, which non-circulant graphs first of all allow for balanced states and what characterizes the balanced states is unknown. Here, we derive rules of how to build non-circulant, planar graphs allowing for balanced states from the simple cycle graph by adding loops or edges to it. We thereby identify different classes of small planar networks allowing for balanced states. Investigating the balanced states' characteristics, we find that the variance in basin stability scales linearly with the size of the graph for these networks. We introduce the \textit{balancing ratio} as a new order parameter based on the basin stability approach to classify balanced states on networks and evaluate it analytically for a subset of the network classes. Our results offer an analytical description of non-circulant graphs supporting stable, balanced states and may thereby help to understand the topological requirements on oscillator networks under global constraints.
\end{abstract}

\maketitle
\section{Introduction}
\label{sec:intro}
Weakly interacting systems are ubiquitous in nature. They play a crucial role for oscillations and timing in biological systems~\cite{Winfree2010}. Examples include circadian clocks, cell metabolism, chemical oscillations or pacemaker cells. Various models have been developed to describe such interacting systems in terms of coupled oscillators~\cite{Pikovsky2003,Arenas2008}. Here, the most well-studied model is the so-called Kuramoto model due to its tractability~\cite{Acebron2005,Strogatz2000,Kuramoto1984} and due to the possibility to derive Kuramoto-like models on very general grounds~\cite{Hoppensteadt2012,Izhi2006,Kuramoto1984}. 

Research regarding the Kuramoto model has focused mainly on oscillator phase synchronization phenomena~\cite{Acebron2005,Arenas2008,Doerfler2014}. Synchronization phenomena can be observed in a broad variety of oscillator network topologies and are comparably easy to study. Yet, little research has addressed so-called \textit{balanced states}, characterized by a vanishing Kuramoto order parameter.
These states have mainly been of interest in the field of control theory~\cite{Sepulchre2006,Paley2007,Paley2007b,jeanne2005}, where they represent desired states to coordinate autonomous vehicles. Yet, balanced states are of much broader concern as they describe coupled oscillators with global constraints that are preventing phase synchronization. For example, such a global constraint is given by conservation of fluid volume in a tubular network with periodically contracting tube walls.
A living example are the contraction patters on the networks formed by the slime mold \textit{Physarum polycephalum}~\cite{Alim2013}. Fascinatingly, these living networks are dynamic in their topology, which poses the question of how network topology affects phase balanced states.

The only class of network topologies so far known to support stable, balanced states in the Kuramoto model are circulant graphs~\cite{Sepulchre2006,Doerfler2014}. However, many real-world graphs are non-circulant, planar graphs. A dynamical model leading to balanced states independent of the underlying topology was introduced by Scardovi et al.~\cite{Scardovi2007} as a modification to the Kuramoto model, but it introduces an additional dynamical variable thus resulting in a different model. 
Therefore, it is still unknown which nontrivial planar network topologies support balanced states 

Here, we construct planar graphs that support stable, phase balanced states in the Kuramoto model.
Starting from the cycle graph as a prototype for graphs with balanced stable fixed points in the Kuramoto model, we first study the effect of loops added to such a graph and show how these may be chosen to preserve balanced states. A balanced state is only preserved if the loops represent the symmetries of the roots of unity found in many balanced states. Subsequently, we probe how additional edges change the balanced fixed points and identify the possible changes leading to balanced topologies, which are given by edges connecting two vertices of equal phase. In particular, we demonstrate how these two building blocks of balanced graphs may be combined to create a large network of Kuramoto oscillators having stable, balanced fixed points. For the graphs created this way, we then show that the variance of basin stability in terms of the winding numbers at stable fixed points scales linearly with the size of the graph. Finally, we introduce the \textit{balancing ratio} as a new measure that allows to compare different graphs in terms of their balanced fixed points using the basin stability approach. We are able to derive analytical scaling laws for this parameter for the examined balanced graphs using the scaling calculated for the variance. Thereby, we manage to capture the effect of loops or additional edges on balanced fixed points in cycle graphs in a quantitative manner. Our results show that topologies other than circulant graphs may support balanced fixed points and offer a new way of looking at balanced topologies.
\section{Methods}
\subsection{Theoretical methods - The Kuramoto model on complex networks}
\label{sec:kuramodel}
In order to study the Kuramoto model on networks, we will first introduce some tools from graph theory \cite{Godsil2001,Diestel2017}. Consider a graph $G(E,V)$ consisting of $N$ vertices with vertex set $V$ and edge set $E$. Then one may define its adjacency matrix $\bm{A}\in\mathbb{N}^{N\times N}$ by the components $A_{ij}$ representing the number of edges starting in vertex $i$ and ending in vertex $j$. Furthermore, one can define the \textit{incidence matrix} $\bm{B}\in\mathbb{Z}^{N\times N_E}$, where $N_E=\abs{E(G)}$ is the number of edges in $G$. It is defined by its components as 
\begin{align*}
B_{ij}=
\begin{cases}
1,&\text{if}~(j,i)\in E(G), \text{ i.e., } e_j\rightarrow v_i \\
-1,&\text{if}~(i,j)\in E(G), \text{ i.e., } e_j\leftarrow v_i\\
0,&\text{otherwise}
\end{cases},
\end{align*}
 where $e_j\rightarrow v_i$ denotes the fact that edge $j$ is incident in vertex $i$ and vice versa.
 Furthermore, the \textit{Laplacian matrix} $\bm{L}\in\mathbb{Z}^{N\times N}$ of a graph $G$ is defined by \begin{align*}
 \bm{L}=\bm{B}\bm{B}^T.
\end{align*}
A graph is called \textit{circulant}, if its Laplacian matrix is a circulant matrix. Circulant matrices are defined by the fact that each row is a cyclic permutation of the previous one. The eigenvalues $\lambda_j$ and eigenvectors $\bm{v}_j, j\in\{1,...N\}$ of the Laplacian matrix for circulant graphs take the simple form of the roots of unity which reads~\cite{Gray2006},
\begin{align}
\lambda_j&=\sum_{k=1}^N l_k (\rho_j)^k,\nonumber\\
\bm{v}_j&=\left(1,\rho_j,\rho_j^2,...,\rho_j^{N-1}\right)^T,~j\in\{0,...,N-1\},
\label{eq:eigcirc}
\end{align}
where $\rho_{j,N}=e^{i2\pi j/N}$
are the $N^\text{th}$ roots of unity, i.e.~the complex numbers such that $(\rho_{j,N})^N=1,~\forall j$. 

In this publication, we focus exclusively on so-called \textit{planar} graphs. A graph is called planar if there exists a drawing of the graph in which no two edges cross. A particularly simple example of a circulant planar graph is given by the \textit{cycle graph} or \textit{simple cycle}, which we refer to by $C_N$ for a cycle graph with $N$ vertices. A cycle is a path that starts and ends in the same vertex without passing through any other vertex twice. If the whole graph is given by a single cycle, it is referred to as a cycle graph.

We are now ready to express the equation of motion for the Kuramoto model on a graph $G$ with $N$ vertices as
\begin{align*}
\dot{\bm{{\vartheta}}}&=\bm{\omega}-K \bm{B}\sin(\mathbf{B}^T \bm{\vartheta}),
\end{align*}
where $\bm{B}$ is the graph's incidence matrix. Here, $\bm{\vartheta}=(\vartheta_1,...,\vartheta_N)^T\in T^N$ are the phase variables evolving dynamically on the $N$-torus $T^N$ over time, $\bm{\omega}=(\omega_1,...,\omega_N)^T\in\mathbb{R}^N$ is a vector of frequencies that is typically drawn from some frequency distribution $g(\bm{\omega})$ and $K$ is the coupling constant determining the strength of the mutual interaction between different oscillators. We assume this coupling constant to be identical for each interaction. In general, one may shift to a rotating frame such that $\bm{1}^T\bm{\omega}=N\cdot\langle\omega\rangle=0$. 


Our main focus will be the Kuramoto model with zero-frequency vector $\bm{\omega}=\bm{0}$. In this case, one may simply assume a coupling constant equal to unity $K=1$ which can
be achieved by rescaling time by the coupling constant $t^\prime=K\cdot t$. The new dynamics will then potentially take shorter or longer times to reach a stable fixed point, but will still follow the same dynamics. This leaves us with the following simplified equation of motion
\begin{align}
\dot{\bm{{\vartheta}}}&=-\bm{B}\sin(\mathbf{B}^T \bm{\vartheta}).
\label{eq:kurareduced}
\end{align}
We will be interested in \textit{fixed points} of this dynamics which are characterized by a vanishing time derivative of the phase variables $\dot{\bm{{\vartheta}}}=\bm{0}$. Identifying such fixed points which are balanced points at the same time is the main goal of this work.

For coupled oscillator models, one can define a measure of synchrony of oscillators as already introduced by Kuramoto~\cite{Kuramoto1984} which is given by 
\begin{align}
  \label{eq:orderparam}
p(\bm{\vartheta})=R(t)\cdot e^{i\langle \vartheta(t)\rangle}=\frac{1}{N}\sum_{j=1}^Ne^{i\vartheta_j(t)},
\end{align}
where $\langle \vartheta(t)\rangle=\frac{1}{N}\sum_{i=1}^N\vartheta_i(t)$ denotes the average angle. Typically, the length of this complex vector $R(t)=\abs{p(\bm{\vartheta})}$ is used as an order parameter for synchrony. This order parameter assumes values between zero for \textit{balanced states} \cite{Sepulchre2006,Sepulchre2005}, also termed \textit{incoherent}~\cite{Strogatz1994}, and one for complete phase synchronization.

Based on this order parameter, we define the set of balanced  states $\mathbb{B}(N)$ for a given number of oscillator $N$ by 
\begin{align}
\mathbb{B}(N)&:=\left\{\bm{\vartheta}\in\mathbb{T}^N\big|\sum_j^N e^{i\vartheta_j} =0\right\}
\label{eq:balanced}.
\end{align} 
Consequently, a planar graph on which the Kuramoto dynamics in Eq.(\ref{eq:kurareduced}) has a stable fixed point that is a balanced state will be called a \textit{balanced graph} in the following. The interplay between the existence and stability of fixed points in the Kuramoto model on the one hand - which is highly dependent on the underlying topology - and the state being balanced on the other hand - which is a requirement unrelated to the topology - will be the main interest of this work.

An observable that we will make use of in the following sections is the \textit{winding number} $q$. Assume that an orientation was assigned to the planar graph underlying the model resulting in an oriented graph $G^\sigma$ and its cycle basis, i.e.~a basis of the graph's cycles space consisting only of simple cycles \cite{Diestel2017}, is given by $\mathfrak{B}_{\mathcal{C}}=\{C_1,...,C_M\}$. 
Then the winding number for some cycle $C_k$ in the graph reads
\begin{align}
\label{eq:windingnumber}
q_{k}=\frac{1}{2\pi}\sum_{(i,j)\in C_k}\Delta_{ij},
\end{align}
where $(i,j)$ are edges composing the cycle and $\Delta_{ij}=\vartheta_i-\vartheta_j$ is the phase difference along that edge taken modulo $2\pi$ to project it onto the interval $\Delta_{ij}\in(-\pi,\pi]$. This winding number $q_k\in\mathbb{Z}$ assumes integer values since all phases along the path following the cycle need to be uniquely defined and in order to retrieve the phase at the starting point, the overall phase difference needs to total to zero modulo $2\pi$.
\subsection{Stable fixed points of the Kuramoto dynamics - the basin stability approach}
In order to be able to study stable balanced fixed points, we 
will classify a fixed point's stability based on the concept of \textit{basin stability} introduced in Ref.~\onlinecite{Menck2013}. It is based on the fixed point's \textit{basin of attraction} $\mathcal{B}$, which is the set of all initial conditions from which the system converges to the fixed point. The basin's volume may then be interpreted as the  probability of returning to the fixed point after a perturbation~\cite{Wiley2006,Menck2013}. In case of the Kuramoto model, the overall phase space volume reads $V=[0,2\pi]^N$. When estimating the basin of attraction $\mathcal{B}(\bm{\vartheta_0})$ of some fixed point $\bm{\vartheta_0}\in V$, the fixed point's indicator function needs to be evaluated in the whole phase space. It reads for some $\bm{\vartheta^\prime}\in V$
\begin{equation*}
\bm{1}_{\mathcal{B}(\bm{\vartheta_0})}(\bm{\vartheta^\prime})=
\begin{cases}
1, \text{ if } \bm{\vartheta^\prime}\in\mathcal{B}(\bm{\vartheta_0})\\
0, \text{ if } \bm{\vartheta^\prime}\not\in\mathcal{B}(\bm{\vartheta_0})
\end{cases}.
\end{equation*}
The basin stability $\mathcal{S}_{\mathcal{B}}(\bm{\vartheta_0})$ of some fixed point $\bm{\vartheta}_0\in\mathbb{T}^N$ is then simply given by the integral over this indicator function normalized by the overall phase space volume 
\begin{equation}
\mathcal{S}_{\mathcal{B}}(\bm{\vartheta_0})=\frac{\int_V \bm{1}_{\mathcal{B}(\bm{\vartheta_0})}(\bm{\vartheta^\prime})\diff^N \bm{\vartheta}^\prime}{(2\pi)^N}\in[0,1].
\label{eq:basinstability}
\end{equation}
A different notion typically used to find stable fixed points in dynamical systems is linear stability analysis.
It is based on the eigenvalues of the dynamical system's Jacobian matrix~\cite{khalil1996}. Based on this notion, one may calculate the following sufficient criterion for a fixed point to be linearly stable in the Kuramoto dynamics independent of the underlying topology~\cite{Jadbabaie2004,Sepulchre2006} 
\begin{align}
\cos(\Delta\vartheta_e)>0,~\forall e\in E(G),
\label{eq:fpstability}
\end{align}
where $\Delta\vartheta_e$ is the phase difference along some edge $e$ in the graph respecting its orientation.
\subsection{The roots of unity as a special class of balanced states}
\label{sec:phasebalancing}
Little knowledge is available about the mathematical structure of balanced states~\cite{Strogatz1994,Paley2007}.
Here, we will discuss particular solutions for balanced states formed by the roots of unity.

Consider a vector of phase variables forming the $N^\text{th}$ roots of unity $\bm{\vartheta}=(\vartheta_1,..,\vartheta_N)^T$, where 
$$\vartheta_j=\frac{2\pi j}{N},~j\in\{1,...,N\}.$$
A well known property of the roots of unity $\rho_{j,N}=e^{\vartheta_j}$ is the fact that they sum up to zero thus making the corresponding angles balanced states, see section~\ref{sec:rootssummation} in the appendix for a proof. Since this is true for all roots of unity, one may choose as well different subgroups of different roots of unity $\vartheta_j=\frac{2\pi j}{N_k}$, where $\sum_k N_k=N$, without changing the fact that they sum up to zero and are thus solutions of balanced states \cite{Alim2013}. Note that for the particularly simple case where $N=N_k$, i.e.~the case of evenly distributed oscillators, this set of angles is also referred to as \textit{splay states} in the literature and plays an important role in neuroscience~\cite{Hadley1987,Nichols1992}.

Importantly, the overall state remains balanced if one adds potentially time-dependent angles $\alpha_k(t)\in\mathbb{S}^1$ to each group of roots of unity summing up to zero separately 
$$\vartheta_j^*=\frac{2\pi j}{N_k}+\alpha_k(t).$$ 
Note that it is possible to choose different $\alpha_k$ for each group of $N_k^\text{th}$ roots of unity. Therefore, a subset $\mathbb{B}_R(N)$ of the set of balanced states reads
\begin{align}
 \mathbb{B}_R(N)&=\Big\{\bm{\vartheta}\in\mathbb{T}^N\big|\vartheta_j=\frac{2\pi j}{N_k}+\alpha_k(t),\label{eq:generalcom} \\
 &\sum N_k=N, N_k \in \mathbb{N},\alpha_k(t)\in\mathbb{S}^1\Big\}\subset\mathbb{B}(N).\nonumber
\end{align}
For small $N<5$, all balanced states may be created this way. The $N_k$ may in general either be constructed using the prime factors of the overall number of oscillators $N$, or any combination of prime numbers summing up to $N$ according to a theorem by Lam and Leung~\cite{LAM2000}.
We will show how these solutions may be realized on a graph-level.
\subsection{Computational methods - Monte Carlo method and cycle flows allow to identify stable fixed points}
\label{sec:methods}
In this work, we make use of two different methods to analyze the Kuramoto model's fixed points numerically. 

On the one hand, we use an algorithm presented recently by Manik et al.~\cite{Manik2017} that may be used to find stable fixed points in the Kuramoto model. It is based on the fact that different solutions to the equation characterizing fixed points in the Kuramoto model may only differ by constant flows around the underlying graph's cycles~\cite{Dorfler2013,Delabays2016,Manik2017}. In addition to that, it allows to establish a one-to-one correspondence between fixed points and winding numbers in the case where all phase differences between neighboring vertices are contained within a $\pi$-interval centered around the origin $\Delta_{ij}\in[-\pi/2,\pi/2],~i,j\in E(G)$~\cite{Delabays2016,Manik2017}. This also guarantees the fixed points to be stable due to the stability criterion in Eq.~\ref{eq:fpstability}. For many planar graphs, all stable fixed points for the Kuramoto model with all-equal frequencies will have phase differences contained in this interval as established by a theorem due to Delabays et al.~\cite{Delabays2017}. 

On the other hand, we use Monte Carlo sampling to determine a fixed point's basin stability.
The phase space in the Kuramoto model has a volume of $V=[0,2\pi]^N$ for $N$ oscillators. Using the Monte Carlo method, $M=10^5$ initial conditions were drawn at random from a uniform distribution spanning $V$ and the number of initial conditions that converged to each of the stable fixed points was counted to estimate the basin stability $\mathcal{S}_{\mathcal{B}}$ similar to the approach in Ref.~\onlinecite{DeVille2016}. Since this is a repeated Bernoulli experiment, it carries the standard error~\cite{Menck2013}, 
$$\sigma(\mathcal{S}_{\mathcal{B}}(\bm{\vartheta}_0))=\sqrt{(\mathcal{S}_{\mathcal{B}}(\bm{\vartheta}_0)\cdot(1-\mathcal{S}_{\mathcal{B}}(\bm{\vartheta}_0)))/M}$$
where $\bm{\vartheta}_0\in\mathbb{T}^N$ is the fixed point of interest and $\mathcal{S}_{\mathcal{B}}(\bm{\vartheta}_0)$ its basin stability.
One can easily see that this error reaches its maximal value $\sigma(0.5)\approx 2\cdot 10^{-3}$ at a basin stability of $\mathcal{S}_{\mathcal{B}}(\bm{\vartheta}_0)=0.5$ given the number of initial conditions used here  $M=10^5$. 

Integration of the Kuramoto equation of motion was performed using a standard ODE solver contained in the \textsc{scipy} package with stepsize $\text{dt}=0.5$. Convergence to a stable fixed point was ensured by checking that the last change was small and calculating the Jacobian eigenvalues. 
When evaluating the balancing ratio, all fixed point were counted as balanced if their order parameter was smaller than a graph-specific threshold value. This threshold value was calculated based on the order parameter at the smallest non-balanced fixed point for the given graph, which varies a lot between different graphs.

Using these two approaches, we could identify all stable fixed points that occupy a significant amount of the overall basin stability and may thus be found by the Monte Carlo method or have all phase differences contained in the above interval and are therefore identified by the algorithm. In addition to that, we were able to find analytical approximations for the phase differences at stable fixed points which we used to recheck in specific cases that we did not miss stable balanced fixed points, which, however, did not reveal any new balanced fixed points not accounted for by either of the two methods.
\section{Results}
\subsection{Creating non-circulant balanced graphs from cycle graphs}
\label{sec:buildingblocks}
We will build non-circulant balanced graphs starting with the simplest example of planar, circulant graphs, the cycle graph and the 3-regular prism graph. The cycle graph is denoted $C_N$, where the index represents the number of vertices in the graph $N$. A cycle is a path that starts and ends in the same vertex consisting of distinct vertices $V(C_N)=\{v_1,...,v_N\}$ and edges of the form $E(C_N)=\{(v_1,v_2),$ $(v_2,v_3),...,(v_N,v_1)\}$~\cite[p.8]{Diestel2017}. If the whole graph is given by one cycle, we refer to the graph as a cycle graph. For this graph and the version of the Kuramoto equation of motion~\eqref{eq:kurareduced} used here, the phase differences at fixed points of the dynamics may be calculated analytically. These phase differences at stable fixed points read~\cite{jeanne2005}
$$
\Delta\vartheta_{c}(q_c)=\frac{2\pi q_{c}}{N},~q_{c}\in [-\lfloor N_c/4\rfloor,\lfloor N_c/4\rfloor]\subset \mathbb{Z},
$$
where $q_{c}$ are the winding numbers characterizing the fixed point and the domain represents stable fixed points.

The 3-regular prism graphs is denoted $Y_N$, where the index once again represents the number of vertices in the graph. They may be created by connecting two cycle graphs to each other at pairs of vertices. For these graphs, phase differences at stable fixed points take the same form as for the cycle graph with the crucial difference that the winding number $q_Y$ characterizing stable fixed points is now smaller $q_{Y}\in [-\lfloor N_Y/8\rfloor,\lfloor N_Y/8\rfloor]$. This is due to the fact that the number of stable fixed points is limited by the number of vertices in the central cycle which consists of half of the overall number of vertices. 

Importantly, the stable fixed points in these two topologies are all balanced except for the synchronized state. Starting from this observation, we will now construct other, non-circulant balanced graphs.
\subsubsection{Only a few building blocks may be used to create balanced graphs from cycle graphs}
\label{sec:balfromcycle}
\begin{figure}[tb!]
\centering
\includegraphics[width=.5\textwidth]{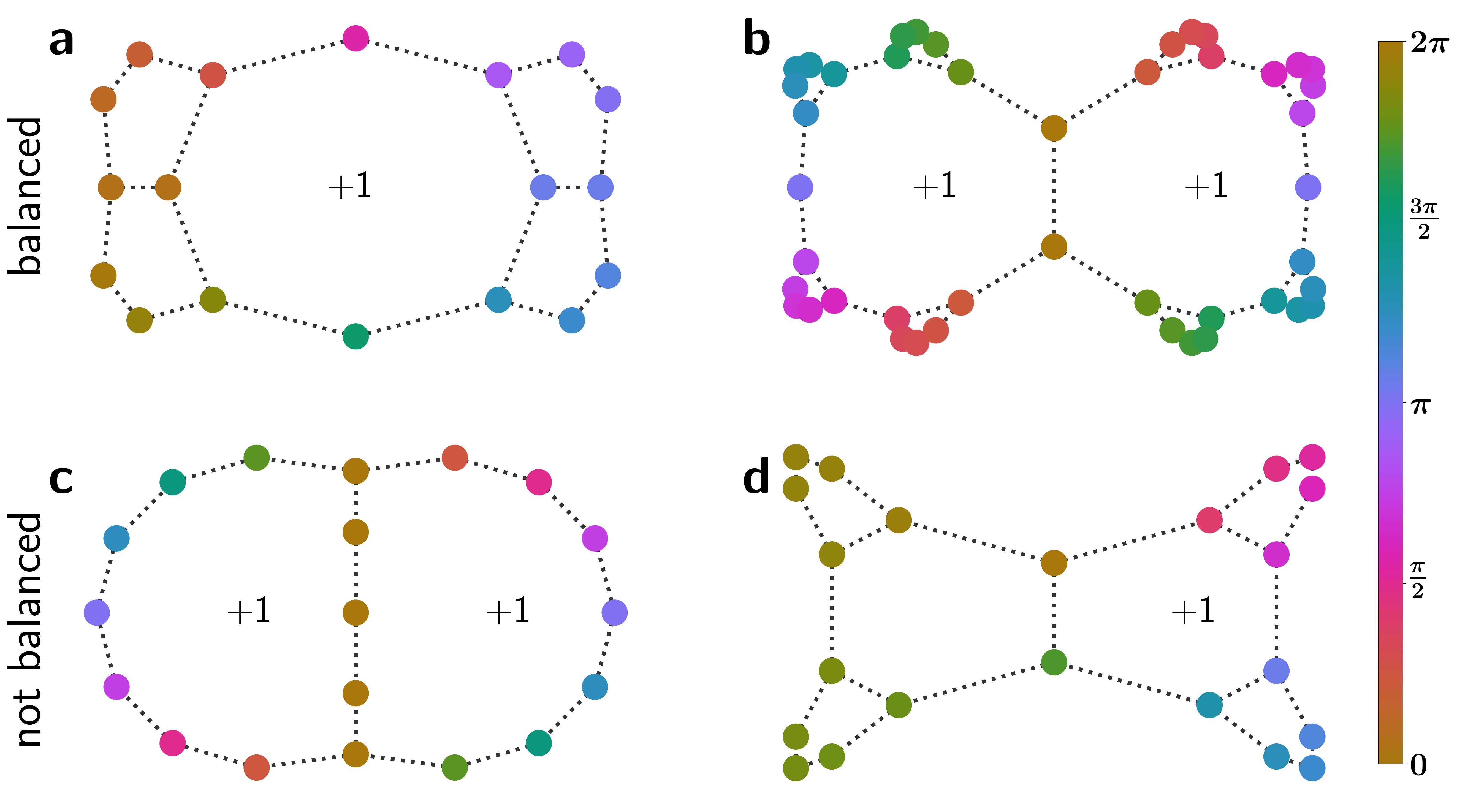}
\caption{Balanced graphs may be constructed by making use of the two different building blocks.
(a) Balanced graph representing the symmetries of the fourth roots of unity with respect to the corresponding winding number $q_1=+1$ in the spectral drawing. (b) Graph consisting of two cycles connected to each other but with additional loops added in such a way, that they are symmetric with respect to each of the big faces with winding numbers
$q_{2,1}=q_{2,2}=+1$. (c) Two cycle graphs connected at more than one edge do no longer possess any stable balanced fixed points. Fixed point shown with $q_{3,1}=q_{3,2}=1$ is balanced if cycles are only connected at a single edge. (d) A graph that is not symmetric with respect to the winding number $q_{4,1}=+1$ of the big face. The stable state shown is not balanced. Faces with no winding number indicated carry winding number zero. All graphs shown here have automorphism groups with basis elements being rotation and reflection along horizontal axis.}
\label{fig:balex}
\end{figure}
As we have seen, all circulant graphs are balanced graphs including the simple cycle. For this reason one may try to construct other balanced graphs by modifying the cycle graph. However, most modifications will make the graphs non-balanced. In Fig.~\ref{fig:balex}, two examples of balanced and non-balanced graphs created from balanced graphs are shown to highlight the difficulties of finding appropriate modifications. The graphs shown in panels c and d are non-balanced whereas the graphs in panels a and b are although their overall symmetries in terms of graph automorphisms are the same. They are given by elementary rotation and reflection along the horizontal symmetry axis.

Here, we will show how one may add edges to cycle graphs, or add loops to them such that the resulting graphs are non-circulant, but still balanced graphs. To achieve this, we will make use of our knowledge about a special class of balanced states constructed from the roots of unity as defined in Eq.~\ref{eq:generalcom}. In general, we found two different ways of modifying cycle graphs that lead to balanced graphs which we will discuss in the following paragraph.
\subsubsection{Loops added symmetrically to a cycle graph may result in a balanced graph}
\label{sec:loopycyclegraph}
\begin{figure}[tb!]
\centering
\includegraphics[width=.5\textwidth]{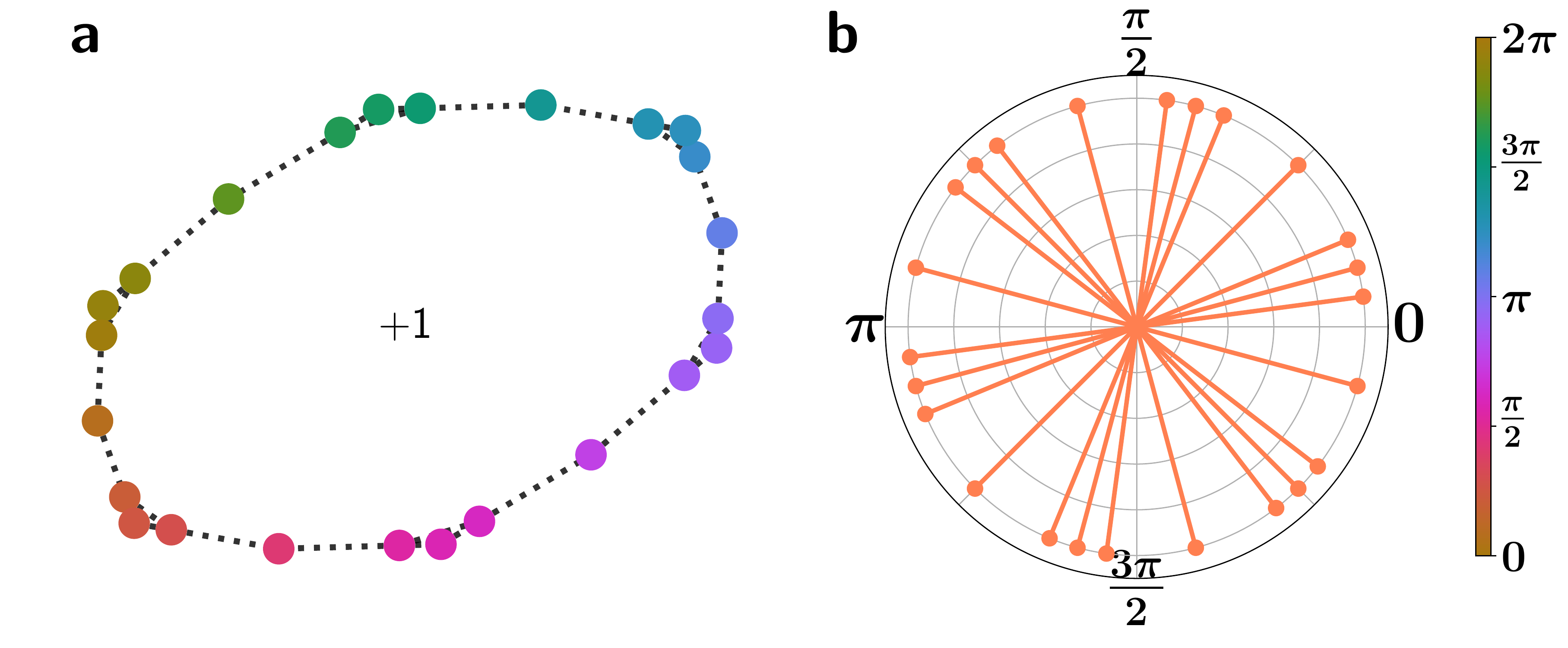}
\caption{(a) The cycle graph $C_{18}$ to which $k=6$ loops of size $m=1$ have been added with angular variables showing a stable, phase balanced fixed point (colorbar). The fixed point has winding number $q=1$ in the central cycle. (b) Angular variables corresponding to the fixed point shown in (a) plotted on the unit circle. The angles in each of the loops are visible and form roots of unity themselves resulting in a balanced fixed point.}
\label{fig:loopcycle_example}
\end{figure}
An obvious way to modify cycle graphs is given by adding loops consisting of several vertices to the graph. Yet, most such operations will destroy the balancing properties of fixed points. The only way to add loops to the graph preserving this property is a symmetric addition of loops reflecting the symmetries of balanced states. For example, if one was to add $k$ loops consisting of $m$ vertices each to a cycle graph with a total of $N$ vertices in such a way that the resulting graph will be balanced, we found that one may in general only do so by choosing $k$ as an integer divisor of $N$ $k|N=n \in\mathbb{N}$. In this case, one may choose the loops to be equidistant in terms of the vertices between them, which corresponds to the roots of unity in angular space. An example representing the $6^\text{th}$ roots of unity on the cycle graph $C_{18}$ is shown in Fig.~\ref{fig:loopcycle_example}. Panel (a) shows the cycle graph to which $k=6$ loops consisting of $m=1$ vertex each have been added and the corresponding color code represents phase variables at this stable fixed point of the Kuramoto dynamics. Panel (b) shows the corresponding angular variables on the unit circle at this fixed point, which makes the one-to-one correspondence between topology and the roots of unity characterizing it obvious.
\subsubsection{Adding edges to balanced graphs may preserve balanced states}
\label{sec:connect_balanced}
\begin{figure}[tb!]
\centering
\includegraphics[width=.5\textwidth]{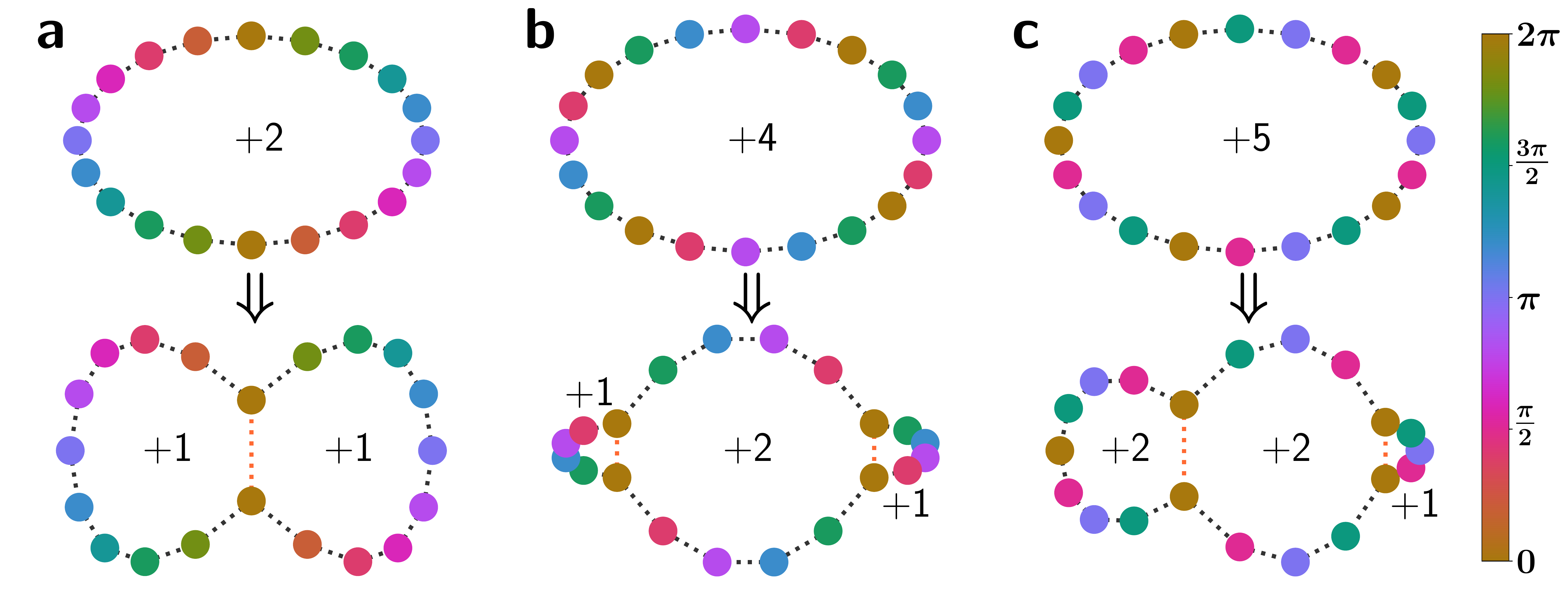}
\caption{Additional edges (red, dotted lines) added to cycle graphs may not change the overall stability properties of stable, balanced fixed points. (a) A  stable, balanced fixed point in a cycle graph (top) with winding number $q_1=+2$ (center of graph) does not loose stability upon addition of an edge connecting two vertices having the same angular value (bottom). The overall winding number is distributed equally over the two resulting graphs (bottom) $q_{1,1}=q_{1,2}=+1$ (center of the graphs). (b) The same procedure may be applied for higher order fixed points (top) by connecting again vertices of equal phases (bottom). The overall winding number $q_2=+4$ is again distributed on the different resulting graphs with $q_{2,1}=+1, q_{2,2}=+2$ and $q_{2,3}=+1$ from left to right. (c) The same mechanism can be applied starting from winding number $q_3=+5$ (center) and connecting vertices of same phase. The last graph's automorphism group has only one element which is the reflection along the horizontal axis, whereas all other graphs have an additional rotational symmetry.}
\label{fig:cyclebalanced}
\end{figure}
Suppose a balanced graph $G$ was found along with the corresponding stable balanced fixed point $\bm{\vartheta}$ characterized by some winding number with absolute value larger than one, $|q|>1$, in one of the graph's cycles. 
If there are two vertices that have exactly the same phase variable $\vartheta_i=\vartheta_j$ then one may add an edge to the graph connecting these two vertices $E(G^\prime)=E(G)\cup\{(i,j)\}$. The corresponding fixed point in the new graph $G^\prime$ will still be a stable, balanced fixed point. 

One may see that the new fixed point is still stable as follows; suppose that $\bm{J}(\bm{\vartheta})\in\mathbb{R}^{N\times N}$ is the Jacobian matrix characterizing the fixed point in the old system before adding an edge with eigenvalues $\lambda_k(\bm{J})<0,~\forall k\in\{1,...,N\}$, which guarantees the fixed point's linear stability by using linear stability analysis~\cite{khalil1996}. 
If the edge added to the graph reads $(i,j)\in E(G^\prime)$, then the Jacobian at the new fixed point $\bm{J}^\prime(\bm{\vartheta}^\prime)$ differs from the old one by a matrix $\bm{J}^\prime(\bm{\vartheta}^\prime)=\bm{J}(\bm{\vartheta})+\bm{j}$ which has the trivial entries $\bm{j}_{kl}=1\cdot(\delta_{ki}\delta_{lj}+\delta_{kj}\delta_{li})-1\cdot(\delta_{ki}\delta_{li}+\delta_{kj}\delta_{lj})$, i.e.~a symmetric matrix with non-zero entries only at the entries characterizing the new edge. The entry unity is due to the fact that the new edge has zero phase difference such that $\cos(0)=1$. This matrix has only one non-zero eigenvalue $\lambda({\bm{j}})=-2$. Using Weyl's inequality for the eigenvalues of symmetric matrices~\cite{Weyl1912}, one can see that the Jacobian eigenvalues $\lambda_k^\prime$ at the new fixed point fulfill $\lambda_k^\prime\leq\lambda_k+0$ due to the fact that zero is the maximal eigenvalue of the additional matrix $\bm{j}$. Thus, all Jacobian eigenvalues at the new fixed point are negative as well and the fixed point is stable.

On the other hand, the fixed point stays balanced since the new fixed point has exactly the same phase variables $\bm{\vartheta}^\prime=\bm{\vartheta}$ and the balancing condition is independent of topology. An example is shown in Fig.~\ref{fig:balex},b, where an edge was added to a cycle graph with loops. Using this procedure, we made use of the fact that one may subdivide the overall balancing property into individual subgroups summing to zero separately, see section~\ref{sec:phasebalancing}.

If the underlying graph is a cycle graph, this procedure may be applied
if the winding number is an integer divisor of $N$ $q|N=m\in\mathbb{N}$. Then there are groups of $q$ vertices having the same angular value since $q=N\cdot \varphi/(2\pi)$, where $\varphi=\vartheta_{i}-\vartheta_{i-1},~\forall i$, implies that $\vartheta_i=\vartheta_{i+m}\operatorname{mod} 2\pi$, i.e.~each angular value is repeated at every $m$-th position.
 It is important to note that although the addition of edges as described above does not change the stability property of the corresponding fixed point it will affect the basin stability of a fixed point as new fixed points are created by adding cycles to the graph. In Fig.~\ref{fig:cyclebalanced}, examples of cycle graphs to which edges have been added without destabilizing the corresponding balanced fixed points are shown. Balanced fixed points in the cycle graph $C_{20}$ characterized by winding numbers $q=2$ (a), $q=3$ (b) and $q=5$ (c) do not become unstable if edges are added connecting vertices having the same phase variable (red dotted lines). The winding number is distributed in correspondence with the number of vertices to the new cycles in the modified graphs. 

Importantly, although all graphs shown here have non-trivial automorphism groups, the two building blocks discussed here may also be combined to create graphs with no trivial symmetry properties, see section~\ref{sec:nonsymbal} in the appendix.
\subsection{Towards a general classification of balanced graphs using the balancing ratio}
\label{sec:results}
Now that we showed how a class of balanced graphs may be constructed, we are going to quantify how balanced they are. In general, there is little knowledge available about the structure of balanced fixed points in balanced graphs apart from the fact that the Kuramoto order parameter vanishes at such points. To be able to compare different balanced graphs in an easily accessible way, we introduce the \textit{balancing ratio} as a new measure based on the basin stability approach. We define this as \textit{the fraction of basin stability occupied by all balanced states in a graph $G$ taken together} 
\begin{equation}
b(G)=\sum_{\bm{\vartheta}\in\mathbb{B}_G(N)}\mathcal{S}_\mathcal{B}(\bm{\vartheta})\in[0,1].
\label{eq:BalancingRatio}
\end{equation}
Here, $\mathbb{B}_G(N)$ is the set of balanced states consisting of $N$ angles, now referring exclusively to the stable balanced states in the graph $G$ as indicated by the subscript and defined by Eq.~\ref{eq:balanced}. This measure may be understood as the probability of ending up in a balanced state when starting from a randomly chosen state. Importantly, most graphs will have vanishing balancing ratio $b(G)=0$ since they are not balanced. In practice, the balancing ratio may be calculated by performing repeated Monte Carlo experiments and counting the number of times a balanced state is reached.
\subsubsection{Variance of basin stability scales linearly with number of vertices for many balanced topologies}
\begin{figure}[tb!]
\centering
\includegraphics[width=0.5\textwidth]{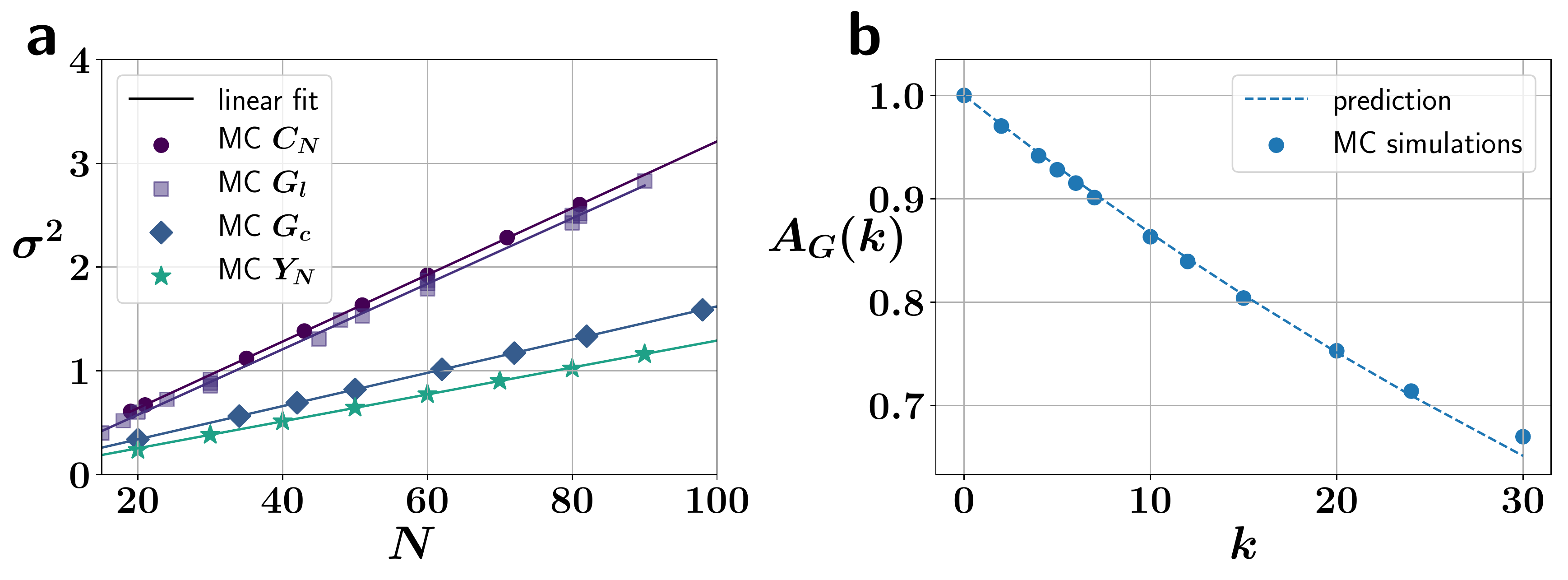}
\caption{(a) Scaling of the variance $\sigma^2$ characterizing the Gaussian distribution describing the basin stability with the number of vertices $N$ in different graphs. Linear fits (dotted lines) and data points from Monte Carlo sampling (markers) are shown for cycle graphs $C_N$ (dark purple dots), cycle graphs with loops $G_l$ (light purple squares), a cycle graph with a central edge added $G_c$ (blue diamonds) and prism graphs $Y_N$ (green stars). Fit parameters for linear relationship $\sigma^2(N)=m\cdot N +b$ are $m_{C_N}=3.2\cdot 10^{-2}, b_{C_N}= -1.7\cdot 10^{-3}$, $m_{G_l}=3.1\cdot \cdot 10^{-2}, b_{G_l}= -4.9\cdot 10^{-2}$, $m_{G_c}=1.6\cdot N\cdot 10^{-2}, b_{G_c}= -1.1\cdot 10^{-3}$ and $m_{Y_N}=1.3\cdot 10^{-2}, b_{Y_N}= -1.0\cdot 10^{-2}$ from top to bottom. Note that for $G_l$, fit was performed in terms of the number of vertices in the central cycle.
(b) Scaling of share of basin stability occupied by the stable fixed points with all winding numbers in outer loops equal to zero $A_G(k)$ in dependence of the number of loops $k$ added to a cycle graph. Data points from Monte Carlo sampling represented by dots and analytical approximation as dotted line are shown. Errors in both figures take a maximal value of $\sigma(\mathcal{S}_\mathcal{B}(\bm{\vartheta}))=2\cdot 10^{-3}$ for each dot and are thus too small to be visible in the figure.
}
\label{fig:sigmacycles}
\end{figure}
In order to calculate this balancing ratio of a circulant graph for an arbitrary number of vertices in the graph, we need to derive an expression for the basin stability of the graph's fixed points.
The basin stability $\mathcal{S}_\mathcal{B}(q_c)$ in dependence of the winding number $q_c$ for the fixed points in the Kuramoto model on a cycle graph is known to follow a Gaussian distribution $\mathcal{S}_\mathcal{B}(q)=\frac{1}{\sqrt{2\pi\sigma^2}} e^{-q^2/(2\sigma^2)}$ and different explanations for this scaling have been suggested~\cite{Wiley2006,Delabays2017b}. 
 As we found here, the variance $\sigma^2_{C_N}(N)$ of this Gaussian distribution scales to good approximation linearly with the number of vertices $N$ as shown in Fig.~\ref{fig:sigmacycles}a (dark purple line), i.e.,
 $$\sigma^2_{C_N}(N)=3.2\cdot 10^{-2}\cdot N -1.7\cdot 10^{-3}.$$
 This fit was calculated by fitting a linear function to the results obtained using Monte Carlo sampling and calculating the variance in the winding number distribution \footnote{Alternatively, one might calculate the variance by fitting the integral under the Gaussian $\mathcal{S}_\mathcal{B}(q)\stackrel{!}{=}(\sqrt{2\pi}\sigma)^{-1}\int_{q-0.5}^{q+0.5}e^{-x^2/(2\sigma^2)} \diff x$ as is done in Ref.~\cite{Wiley2006}. However, we found the given approach to yield better results when comparing it to results from Monte Carlo sampling.}. For moderate values of $N$, the scaling is $\sigma^2_{C_N}(N)=3.2\cdot 10^{-2}\cdot N$ effectively since the y-intercept is negligible. 
 
In addition to the linear scaling observed for the variance in case of the cycle graph, we found this scaling to hold for other graphs as well. For the prism graph $Y_N$, the variance scales as $\sigma^2_{Y_N}(N)=(1.3\cdot N-1.0)\cdot 10^{-2}$ if fitted linearly to the data. This result coincides very well with the data points obtained from simulation, see Fig.~\ref{fig:sigmacycles},a, green stars. 

Another example for which we found the linear scaling to hold is the graph created by adding a single edge to a cycle graph thus creating two cycles of the same size connected at a single edge as described in section~\ref{sec:connect_balanced} and shown in Fig.~\ref{fig:cyclebalanced} a, bottom. In this case, the basin stability follows a two-dimensional Gaussian scaling $$\mathcal{S}_\mathcal{B}(q_1,q_2)=\frac{1}{2\pi\sigma_1\sigma_2} e^{-1/2((q_1/\sigma_1)^2+(q_2/\sigma_2)^2)},$$
if plotted against the winding numbers of the two cycles $q_1$ and $q_2$. If the two cycles have the same number of vertices $n=N/2+1$, the Gaussian distribution is well-described by a single standard deviation $\sigma_1=\sigma_2=\sigma$. This variance shows the same scaling as for the single cycle if plotted against the number of vertices in each of the cycles $n$. In terms of the number of vertices $n$ in a cycle, it reads $$\sigma^2_{G_c}(n)=(3.2\cdot n-1.1)\cdot 10^{-2}.$$
This scaling is shown in terms of the overall number of vertices $N$ in Fig.~\ref{fig:sigmacycles} a (blue).

Finally, this scaling is also valid for graphs to which loops have been added according to the scheme described in section~\ref{sec:loopycyclegraph}. The variance of the Gaussian scaling observed in terms of the winding number was evaluated for the winding number in the central loop. In terms of the number of vertices in this central loop $N^*=N-k\cdot m$, where $N$ is the overall number of vertices in the graph, $k$ is the number of loops added and $m$ is the number of vertices in each loop, the overall variance shows a similar scaling compared to the cycle graph
$$
\sigma^2_{G_l}(N^*)=(3.1\cdot N^*-4.9)\cdot 10^{-2}.
$$
This scaling was obtained using a linear fit on the Monte Carlo results of $18$ different graphs with $k\in\{2,3,5\}$ loops and different number of vertices. It is shown in Fig.~~\ref{fig:sigmacycles} a (purple).
\subsubsection{Balancing ratio for circulant graphs shows a simple square root scaling}
Inspired by the numerical results, we will now proceed to show how one may calculate the balancing ratio in circulant graphs.
Importantly, all states except for the synchronized state with winding number $q_c=0$ are balanced for circulant graphs. Therefore, the balancing ratio is calculated by summing over all stable states' basin stability except for the synchronized state which is the peak of the Gaussian distribution. This peak is of height $(\sqrt{2\pi}\sigma_{C_N}(N))^{-1}$ for a normalized Gaussian distribution where $\sigma_{C_N}(N)$ is the Gaussian's standard deviation. Using the linear scaling for the variance $\sigma^2_{C_N}(N)$, we get the following expression for the balancing ratio of a cycle graph $C_N$ in dependence of the overall number of vertices in the graph
\begin{align}
b(C_N)&=\sum_{\abs{q}\neq 0}\mathcal{S}_\mathcal{B}(q)=1-\left(\sqrt{2\pi}\sigma_{C_N}(N)\right)^{-1}\label{eq:basstabcycle}\\
&\approx 1-\left(\sqrt{2\pi(3.2\cdot 10^{-2} N)}\right)^{-1}\nonumber.
\end{align} 
Here, we neglected the y-intercept due to its negligible effect on the result.
This approximation is shown for different numbers of vertices $N$ in Fig.~\ref{fig:StabilityAll} (purple). It coincides very well with the results from Monte Carlo sampling.

Along the same lines, one may calculate a similar scaling law for the prism graph $Y_N$. Again plugging in the linear scaling obtained for the variance, the balancing ratio for this graph is calculated to be
\begin{align}
b(Y_N)&=\sum_{\abs{q}\neq 0}\mathcal{S}_\mathcal{B}(q)=1-\left(\sqrt{2\pi}\sigma_{Y_N}(N)\right)^{-1}\label{eq:basstabprism}\\
&\approx 1-\left(\sqrt{2\pi(1.3\cdot N-1.0)\cdot 10^{-2})}\right)^{-1}\nonumber.
\end{align}
Here, $q$ is the winding number in the central cycle of the prism graph.
This approximation is shown in Fig.~\ref{fig:StabilityAll} (dark green line) again coinciding with the results obtained from Monte Carlo sampling.
\subsubsection{Balancing ratio depends on newly created cycles when adding edges to cycle graphs}
Having solved balancing ratios for simple circulant graphs, we now consider the balanced graph obtained by adding a single edge to a cycle graph $C_N$ such that each of the newly created cycles has the same number of vertices $n=N/2+1$. We will refer to this graph by $G_c$. This graph is the most simple balanced graph one may obtain using the above building block of adding edges to balanced graphs and thus provides a natural extension of cycle graphs. Exemplarily, such a graph is shown in Fig.~\ref{fig:cyclebalanced} a, bottom. Making use of the results obtained for the cycle graph, we will show how one may calculate the balancing ratio for this graph. 

Fixed points in this graph are characterized by the winding numbers in the two cycles $\bm{q}=(q_1,q_2)^T$. In this case, the only balanced states that we found were states where the two cycles have the same winding number, i.e.~opposite flows at the shared edge, which are precisely the fixed points described in section~\ref{sec:connect_balanced} and in one-to-one correspondence with the balanced fixed points in the cycle graph without additional edge. Here, we assume the cycles and edges to be both oriented counterclockwise and the edge shared between the cycles to be oriented along the first cycle characterized by $q_1$.
We define the set characterizing all states where the two cycles have identical winding number by
$$Q_c=\{\bm{q}\in\mathbb{D}_c(\bm{q})\big|q_1=q_2,q_1\neq 0\}.$$
Here, $\mathbb{D}_c(\bm{q})=\left[-\lfloor\frac{n}{4}\rfloor,\lfloor\frac{n}{4}\rfloor\right]^2\subset \mathbb{Z}^2$ is the domain where the fixed points are guaranteed to be stable, but not all combinations of winding numbers in this set may necessarily be realized in a stable fixed point. Again, the synchronized state is not balanced for this topology and thus excluded from this set.
Using this definition, the balancing ratio may be calculated as
$$
b(G_c)=\sum_{\bm{\vartheta}\in\mathbb{B}_{G_c}(N)}\mathcal{S}_\mathcal{B}(\bm{\vartheta})=\sum_{\bm{q}\in \mathbb{D}_c(\bm{q})}\mathcal{S}_\mathcal{B}(\bm{q})\delta_{q_1,q_2}-\mathcal{S}_\mathcal{B}((0,0)^T),
$$
where $\delta_{q_1,q_2}$ is the Dirac delta such that $\delta_{q_1,q_2}=1,$ if $q_1=q_2$ and $0$ otherwise. This sum may be evaluated by making use of the fact that the basin stability in terms of the winding numbers follows a two-dimensional Gaussian distribution and the scaling of this Gaussian's variance in terms of the number of vertices $n$ is again linear. This allows us to evaluate the balancing ratio for the graph $G_c$ 
\begin{align}
b&(G_c)=\sum_{\bm{q}\in \mathbb{D}_c(\bm{q})}\frac{1}{2\pi\sigma^2}e^{-(q_1^2+q_2^2)/(2\sigma^2)}\delta_{q_1,q_2}-\frac{1}{2\pi\sigma^2}\nonumber\\
&=\sum_{q\in[1,\lfloor n/4\rfloor]}\frac{1}{\pi\sigma^2}e^{-q^2/\sigma^2}
\label{eq:basstabconnection}\\
&=\sum_{q\in[1,\lfloor n/4\rfloor]}\frac{1}{\pi(3.2\cdot n -1.1)\cdot 10^{-2}}e^{-q^2/(3.2\cdot n -1.1)\cdot 10^{-2}}\nonumber.
\end{align}
This result is shown in Fig.~\ref{fig:StabilityAll}, (dark blue) and agrees very well with the result obtained from Monte Carlo sampling. Comparing this result to the one obtained for the cycle graph $C_N$ in Eq.~\ref{eq:basstabcycle}, one may notice that it is always smaller by a prefactor of $(2\pi)^{-1/2}$. In addition to that, the variance $\sigma^2(N)$ shows a more moderate scaling with the number of vertices for the given graph and the sum scales with $\sigma^{-2}\varpropto N^{-1}$ in the prefactor compared to the $\sigma^{-1}\varpropto N^{-1/2}$ scaling observed for the cycle graph. This explains the different scalings observed in Fig.~\ref{fig:StabilityAll}.
\subsubsection{Balancing ratio for cycle graph with loops depends on graph symmetries}
\label{ap:balratioloopcycle}
In this section, we will show how a similar result may be obtained for cycle graphs to which loops have been added. Consider a cycle graph to which $k$ loops have been added by adding $m=3$ vertices to the graph for each of the loops and connecting them to a single edge of the graph, thus creating loops consisting of $m+2=5$ vertices, see e.g.~Fig.~\ref{fig:SummationBalancing_small} b, for an example with two loops. Each such loop results potentially in the creation of two new stable fixed points characterized by the two non-zero loop winding numbers $q_l\in\{1,-1\}$. For this reason, stable fixed points with all phase differences contained in the interval $[-\pi/2,\pi/2]$ centered around the origin may be uniquely characterized using the winding vector $\bm{q}=(q_1,q_{l,1},...,q_{l,k})^T$, where $q_1$ is the winding number in the central loop and the $q_{l,k}$ are the $k$ winding numbers characterizing the loops.

In order to be able to calculate the balancing ratio for such graphs, we will first study the effects of loops on the overall basins stability occupied by the newly created fixed points. The basin stability occupied by such a stable state in which exactly one loop has a non-zero winding number stays roughly constant over all numbers of vertices and loops with the average value being
$$\mathcal{S}_\mathcal{B}(q_l=1)=\mathcal{S}_\mathcal{B}(q_l=-1)=(7.1\pm 0.3)\cdot 10^{-3},$$
where the given error is the standard deviation over all samples. This value was obtained for cycle graphs $C_{120}$ to which $k\in\{0, 2, 4, 5, 6, 7,10,12, 15, 20,  24,30\}$ loops have been added, but was also confirmed for much smaller cycle graphs $C_{30}$ with loops added where a similar value was obtained. Thus, the probability of finding a non-zero winding number in one of the loops when starting from a random initial condition for the given graphs reads 
\begin{align}
p(\abs{q}_l=1)=2\cdot\mathcal{S}_\mathcal{B}(q_l=1)=(1.42\pm 0.06)\cdot 10^{-2}.
\end{align}
This allows us to estimate the average basin stability occupied by states with a non-zero winding number in dependence of the number of loops $k$ by 
\begin{align}
A_G(k)=(1-p(\abs{q}_l=1))^k,
\label{eq:areanonzero}
\end{align}
due to the fact that the probability of observing a zero winding number in a given loop reads $1-p(\abs{q}_l=1)$. Here, we assumed the probability of observing non-zero winding numbers in different loops to be independent.
\begin{figure}[tb!]
\centering
\includegraphics[width=.5\textwidth]{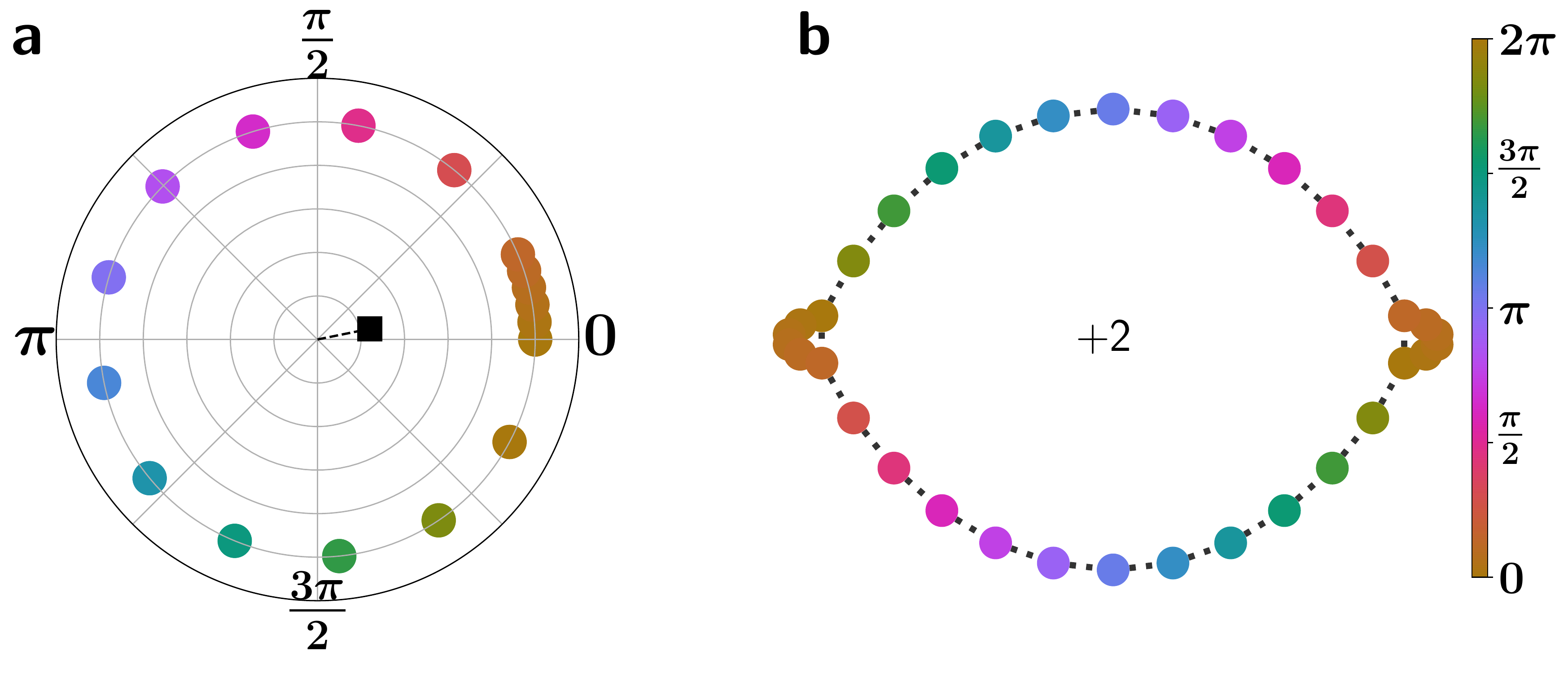}
\caption{Winding numbers multiple of the number of loops induce asymmetries in the phase variables at stable fixed points. (a) Phase variables at the stable fixed point depicted in (b). The asymmetry towards $\vartheta=0$ is clearly visible and indicated by the Kuramoto order parameter in Eq.~\ref{eq:orderparam} (black square), which assumes a non-zero value for this configuration. (b) Cycle graph $C_{28}$ to which $k=2$ loops have been added. The stable fixed point shown here with central winding number $q=2$ is not balanced.}
\label{fig:SummationBalancing_small}
\end{figure}

In order to calculate the balancing ratio for this topology, we need to derive an expression for the basin stability occupied by balanced states. In general, such cycle graphs with loops allow for a lot of different balanced states with non-zero winding numbers in some of the loops according to their symmetry properties, see section~\ref{ap:balratioloopcycle} in the appendix for a detailed discussion. However, in the given setting where loops consist only of a small number of vertices in comparison to the main cycle, the basin stability taken by such states is negligible. For this reason, we will focus on balanced states where all loops have a zero winding number. Compared to the balancing ratio for the cycle graph, basin stability for such graphs will then be reduced by the above factor $A_G(k)$.

In addition to that, the symmetry of the added loops makes a few previously available balanced fixed points now non-balanced ones. This is the case if the winding number is a multiple of the number of loops. To formalize this argument, define the following set characterizing the central loop's winding numbers at balanced fixed points in dependence of the number of loops added to the graph $k$ by 
\begin{align*}
Q_{l,1}(k)=\{q_1\in\left[-\lfloor\frac{N}{4}\rfloor,\lfloor\frac{N}{4}\rfloor\right]\big| \nexists n_1\in\mathbb{N}_{0} \text{ s.t. }\abs{q_1} = n_1\cdot k\}.
\end{align*}
In Fig.~\ref{fig:SummationBalancing_small}, a stable fixed point that is balanced in the corresponding cycle graph, but became non-balanced due to the loops is shown. The figure shows the cycle graph $C_{28}$ to which $k=2$ loops have been added making the stable fixed point with $q_1=2$ now non-balanced. This is due to an asymmetry in the phase variables towards $\vartheta=0$ making the resulting fixed point non-balanced as indicated by a non-zero order parameter, Fig.~\ref{fig:SummationBalancing_small} a, (black square). For this reason, such fixed points need to be excluded when calculating the balancing ratio.

Using these two results, we are now ready to write down an estimate for the basin stability in a graph with loops
\begin{align}
b&(G_l,k)=\sum_{\bm{\vartheta}\in\mathbb{B}_{G_l}(N)}\mathcal{S}_\mathcal{B}(\bm{\vartheta})\approx \sum_{q_1\in Q_{l,1}(k),\bm{q}_l=\bm{0}}\mathcal{S}_\mathcal{B}(\bm{q})\nonumber\\
&= A_G(k)\cdot(1-\sum_{q\in Q_{l,1}}\frac{1}{\sqrt{2\pi\sigma^2}}e^{-q^2/(2\sigma^2)})\label{eq:basstabloops}\\
&\approx (1-1.42\cdot 10^{-2})^k\cdot\left(1-\sum_{q\in Q_{l,1}}\frac{e^{-q^2/(6.4\cdot 10^{-2}\cdot N^*)}}{\sqrt{2\pi \cdot 3.2\cdot 10^{-2}\cdot N^*}}\right). \nonumber
\end{align}
Here, $N^*=N-k\cdot m$ refers once again to the number of vertices in the central cycle without loops.
Analyzing this expression, one may notice the relationship to the basin stability of a simple cycle in Eq.~\ref{eq:basstabcycle}. This expression differs from the one for the simple cycle in terms of a factor accounting for the basin stability in states with non-zero winding numbers in loops and in terms of several winding numbers other than the synchronized state being now excluded from the summations due to the additional symmetries induced by the loops. Importantly, this result will always yield lower values of basin stability in comparison to the corresponding cycle graph.
This approximation is shown in Fig.~\ref{fig:StabilityAll} (light purple) for $k=2,3,5$ along with the corresponding results from Monte Carlo sampling where the number of loops for the respective graph is indicated by a purple number.
\subsubsection{Balancing ratio allows to compare balanced states in different types of balanced graphs}
\begin{figure*}[tb!]
\centering
\includegraphics[width=1.0\textwidth]{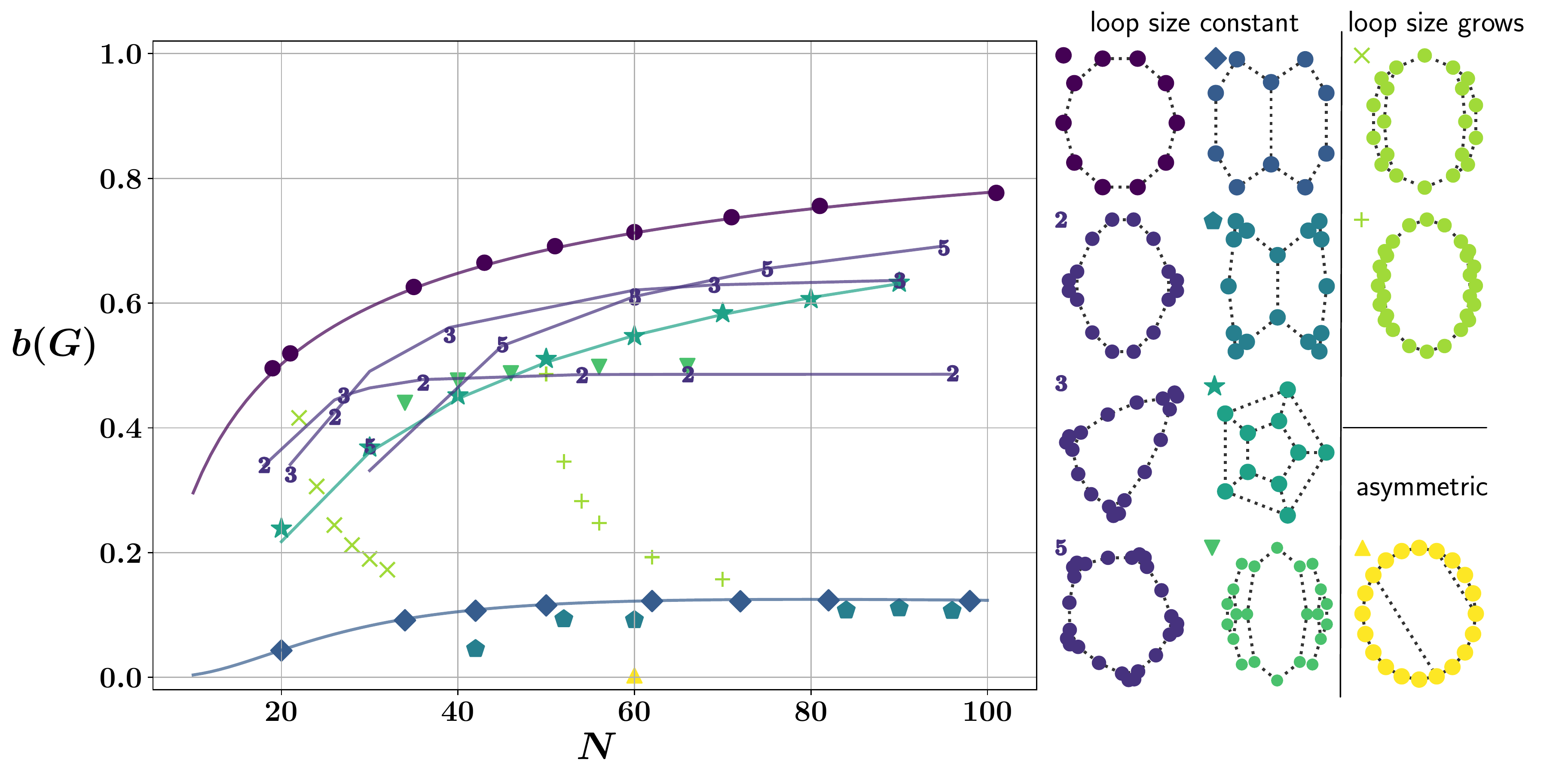}
\caption{Balancing ratio $b(G)$ plotted for different balanced graphs against total number of vertices in graph $N$. Color code and symbols represent different prototypes for the graphs analyzed here as shown in legend, right. For the cycle graph with loops purple numbers $2$, $3$ and $5$ correspond to the number of loops added. For graphs on the left side of the legend, only the central cycle grows with increasing $N$ whereas the added loops - if they exist - stay constant in size. For graphs on the right side of the legend, loops grow accordingly with increasing $N$ (top). Graph at the bottom right of the legend marked by yellow color is the only asymmetric graph analysed here. Results marked by symbols were obtained using Monte Carlo sampling performed for respective graph except for analytical results indicated by colored lines. See text for detailed information on how the respective graph was constructed exactly. Errors take a maximal value of $\sigma(\mathcal{S}_\mathcal{B}(\bm{\vartheta}))=2\cdot 10^{-3}$ for each dot and are thus too small to be visible in the figure.}
\label{fig:StabilityAll}
\end{figure*}
Finally, the balancing ratio may also be calculated numerically using the Monte Carlo method for different types of balanced graphs for which there are no analytical results available so far. To compare the results for those and the previously discussed topologies among each other, we plotted the balancing ratio $b(G)$ against the overall number of vertices in the graph $N$ in Fig.~\ref{fig:StabilityAll}. Different colors and symbols encode different types of graphs as indicated by different prototypes shown in the legend on the right.

The cycle graph appears as an upper bound on the balancing ratio for a given number of vertices $N$ out of all balanced graphs. This might be due to the simple structure of cycle graphs making the balanced states most easily accessible for random initial configurations and the fact that there is only one non-balanced fixed point for this topology. 

Furthermore, we note that all of the graphs analyzed here show a monotonic scaling with the number of vertices in the graph. While most graphs show an increase in the balancing ratio with the number of vertices in the graphs, cycles with long loops added and growing in vertices by growing loop vertices as well show a monotonic decrease with the overall number of vertices,  see symbols $\cross$ and $+$ and light green color in Fig.~~\ref{fig:StabilityAll}. These graphs were constructed by adding two loops of equal lengths $l$ to both sides of a cycle graph symmetrically with the loops spanning $l-1$ of the cycle's vertices such that the additional cycles now share more than one edge with the central cycle. For the graphs represented by symbol $\cross$, loops are of lengths $l\in\{2,3,4,5,6,7\}$ added to the cycle graph $C_{18}$ resulting in graphs of sizes $N\in\{22,24,26,28,30,32\}$, respectively. Other copies of this kind of graph are marked by symbol $+$ and created from the cycle graph $C_{46}$ with higher numbers of vertices $N\in\{50,52,54,56,62,70\}$ and loops of lengths $l\in\{2,3,4,5,8,12\}$ and show similar values of balancing ratios and a similar scaling. These graphs show a clear decrease with the length of the loops added which can be explained by noticing that the length of the loops corresponds to the number of possible fixed points in the graphs due to the increase in cycle length. The new fixed points will occupy some basin stability themselves such that the balanced fixed points occupy less and less basin stability with increasing length of loops. This scaling was confirmed by some preliminary analysis counting the share of fixed points being balanced. The upper bound given by the corresponding cycle graphs $C_{46}$ and $C_{18}$ is reached if the length of the loops is reduced to zero.

Furthermore, the graphs represented by symbol $\blacktriangledown$ and light green color represent cycle graphs to which more complicated structures have been added on two symmetric ends resulting in three additional cycles at both sides. This graph shows the same trend as the simple cycle graphs to which loops have been added, which are shown as light purple numbers. This is due to the fact that these graphs are similar to the latter graphs but different in terms of the number of cycles in the graph.

In terms of the effect additional loops have on balanced states in cycle graphs, it is easily visible that they result in a lower balancing ratio compared to the corresponding cycle, see graphs represented by purple numbers $2$, $3$ and $5$ indicating the number of loops. Counterintuitively, more loops added to the graph might lead to a higher balancing ratio than fewer loops, although the former represent a stronger perturbation to the circular topology. This is due to the fact that more states with lower winding number remain balanced if the number of loops is increased, which take a higher amount of basin stability compared to the states with higher winding numbers. This effect compensates the additional factor $A_G(k)$ for increasing numbers of vertices in the graph leading to higher balancing ratios for the graphs with more loops. However, the reduction in the balancing ratio for any number of loops added symmetrically to the graph is always smaller than adding a central edge to the graph which leads to a much stronger decrease in the balancing ratio, see blue markers $\diamond$. All loops added here consist of $m=3$ additional vertices.

For the graph formed by a cycle graph to which a central edge \textit{and} loops have been added, represented by symbol $\pentagon$ and cyan color, the two small cycles were of size $n\in\{13,20,36\}$ each and $k\in\{2,3,4\}$ loops of size $m=3$ have been added to each of the cycles resulting in the points shown at $N\in\{42,52,60,84,90,96\}$, respectively. In contrast to the results for the single cycle to which loops have been added, the graphs show similar values of balancing ratios as for the corresponding graph created from adding a central edge to the cycle graph. 

Finally, the graph represented by yellow color and symbol $\blacktriangle$ is the only asymmetric graph shown here. This graph is created by adding two edges to the cycle graph $C_{60}$ asymmetrically, thus resulting in a graph with cycles of length $31$, $21$ and $11$. The only balanced fixed points for this graph are characterized by winding number $|q|=6$ in the corresponding cycle graph. It shows a small, non-vanishing balancing ratio of around $b(G)\approx 3\cdot 10^{-3}$.
\section{Summary and Outlook}
\label{sec:discussion}
In this work, we showed how additional loops and edges affect stable, balanced fixed points in the Kuramoto model on cycle graphs. To do so, we used the roots of unity as a particular class of balanced states to construct graphs that support stable balanced states.
We examined the basin stability of stable fixed points in the constructed topologies as a measure of stability and used the results to evaluate the balancing ratio as a new order parameter for balanced graphs. In order to quantify the effect of additional connections on balanced states, we evaluated this parameter both analytically and numerically for different topologies. Our results show that there are numerous non-circulant balanced topologies that may be constructed using simple building blocks. In general, the addition of loops or edges to cycle graphs was shown to worsen the graph's balancability, with the effect of adding a central edge to the balanced graph being much stronger than the one of additional loops.

The balancing ratio introduced here as an order parameter for balanced graphs provides an easily accessible tool to classify stability in balanced graphs. We offer an analytical description of this parameter for a subset of balanced topologies which might help to yield more insight into the structure of balanced fixed points in general. On the other hand, the balancing ratio may be easily calculated for any topology using numerical methods. It also offers the simple interpretation of the probability of ending up in a balanced state. This could be relevant for real-world networks if balanced states describe a desired mode of operation.

Using two basic building blocks, we showed how a large class of balanced graphs may be constructed. In principle, the building blocks shown here might also be used to create graphs representing tilings of two dimensional space but still being balanced. An important task for future work would be to find a classification for the building blocks on more mathematical grounds and to check which planar balanced graphs exist that may not be produced using the above schemes. Furthermore, one might move away from planar graphs and graphs with small degree being the main focus in this work and look for classification schemes for balanced graphs in general. 

We found that the standard deviation of basin stability in terms of the winding numbers scales linearly with the number of vertices in the graph for many topologies. This result yields new insights about the winding number distribution of stable states in cycle graphs and beyond.
However, it is unclear how the prefactor $3.2\cdot 10^{-2}$ found to describe the slope of the variance $\sigma^2$ in terms of the number of vertices in a cycle graph $N$ may be explained theoretically. Neither is there a theoretical explanation for the probability of finding a state with non-zero winding number in one of the loops $p(\abs{q}_l=1)$ nor a scaling of this number with number of loops $k$, size of loops $m$ or number of vertices $N$. One possible way to continue the study of these graphs would be to quantify this probability for different sizes of loops and relate them to the Gaussian scaling observed for connected cycles of the same size or single cycles. On the other hand, it would be interesting to see if the linear scaling of the Gaussian's standard deviation with the number of vertices found here to be valid for two circulant graphs extends to other circulant graphs - or even other topologies. Using this scaling and combining different graphs, it could be possible to analytically calculate the basin stability of single fixed points for much larger graphs.

Although most of the networks constructed show an obvious symmetry in correspondence with the roots of unity, we were not able to relate our findings to results on cluster synchronization in coupled oscillator models in relationship to graph symmetry~\cite{Pecora2014,Schaub2016} or symmetries in the master stability function~\cite{Dahms2012}. It would be an important goal for future works to be able to relate the results found here to graph symmetries in the underlying networks.

The present work is the first one to study planar graphs under global constraints that manifest through phase balancing of oscillators on more general grounds. The building blocks studied here allow to create large networks of oscillators supporting balanced states, providing the mathematical framework to understand constrained oscillators in nature but also to encode control in robotics and autonomous vehicles. The methods introduced here further allow to quantify the stability of the balanced states in such networks - for many networks even analytically. 
\section*{Acknowledgments}
This work was funded by the Max Planck Society.
%

\appendix
\section{Non-symmetric balanced networks}
\begin{figure}[htb]
\centering
\includegraphics[width=0.5\textwidth]{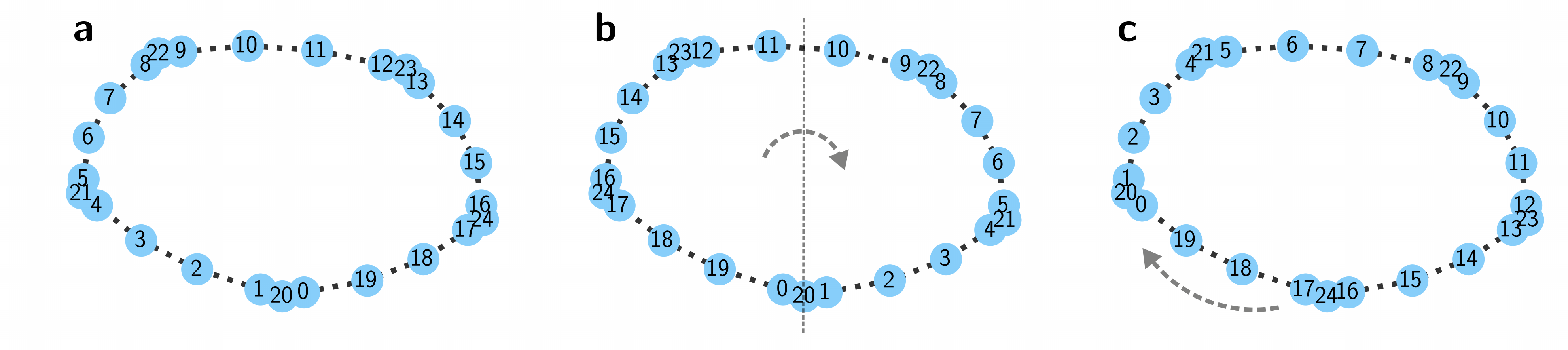}
\caption{Automorphisms for a cycle-like graph with $k=5$ loops symmetrically added to the graph and a total of $N=25$ vertices. Numbers on vertices (light blue circles) label vertices uniquely which are connected through edges (dotted, black lines). The graph displayed (a) is modified by applying two different automorphisms, namely the two basis elements of its automorphism group. (b) The first basis element is a reflection along a symmetry axis of the graph (straight, dotted line, reflection indicated by arrow). (c) The second one is a rotation (grey, dotted arrow) by a length corresponding to the distance between individual loops. Together, these symmetries represent the dihedral symmetries.}
\label{fig:automorphisms}
\end{figure}
\subsection{Graph automorphisms and symmetries}
\label{sec:automorphisms}
An important tool to study dynamical systems evolving on graphs - and in particular balanced states as possible dynamics - are graph symmetries. \textit{Graph automorphisms} in this section being the major tool used for studying symmetries on graphs. We will introduce them following Ref.~\cite[ch. 1]{Diestel2017}. Consider two graphs $G$ and $G^\prime$ with vertex sets $V(G)$, $V(G^\prime)$ and edge sets $E(G)$ and $E(G^\prime)$. A map $\varphi:V\rightarrow V^\prime$ between their vertex sets is a \textit{homomorphism} between the two graphs if it preserves the adjacency of vertices $(v_i,v_j)\in E(G)\Rightarrow (\varphi(v_i),\varphi(v_j))\in E(G^\prime)$. It is an \textit{isomorphism} if the opposite is true as well, i.e.~the map is bijective and its inverse is a homomorphism as well, which implies $(\varphi(v_i),\varphi(v_j))\in E(G^\prime)\Rightarrow (v_i,v_j)\in E(G)$. Finally, an \textit{automorphism} is an isomorphism from $G$ to itself. The set of all automorphisms of a graph forms a group. Since automorphisms are bijective, the combination of two automorphisms is an automorphism once again and the map preserving each vertices' position, i.e.~the identity map, is an automorphism itself which shows their group structure.

The main graphs of interest in this publication are cycle graphs and its modifications. The basis elements forming the automorphism group of the cycle graph with loops added symmetrically are shown in Fig.~\ref{fig:automorphisms}. Panel (a) shows the cycle graph $C_{20}$ to which $k=5$ loops of size $m=1$ have been added symmetrically. Panel (b) and (c) show the basis elements of the graph's automorphism group, namely an elementary reflection along one of the graph's symmetry axis (b) and an elementary rotation corresponding to the distance between different loops (c).
\subsubsection{Non-symmetric balanced graphs}
\label{sec:nonsymbal}
\begin{figure}[tb!]
\centering
\includegraphics[width=.5\textwidth]{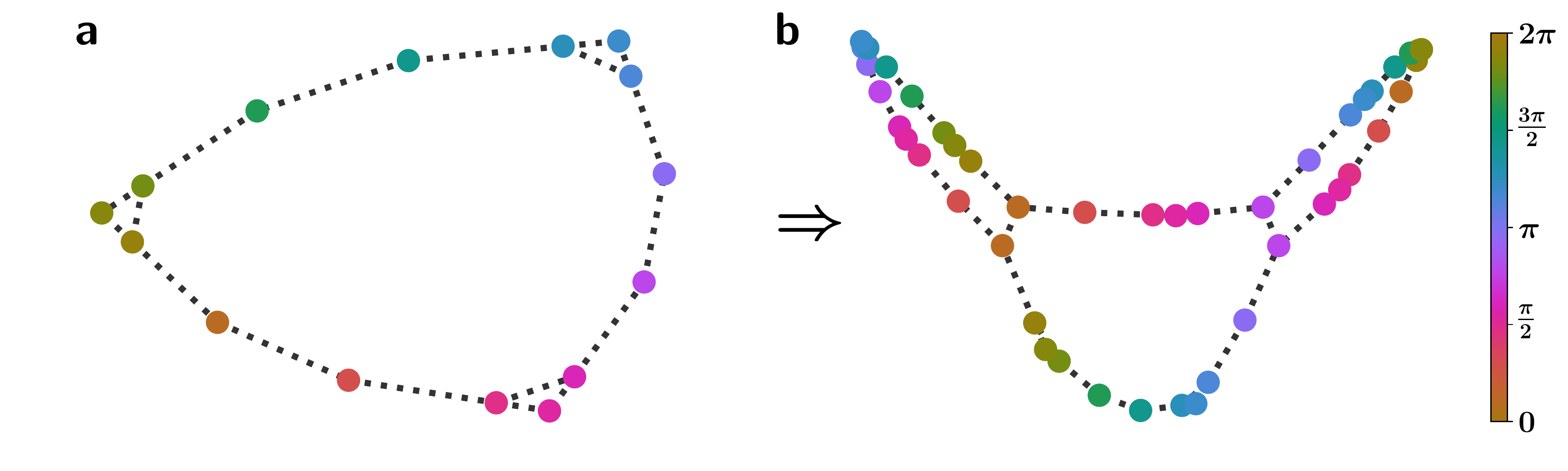}
\caption{A balanced graph with empty automorphism group may be created by combining the two building blocks of adding loops to cycle graphs and connecting copies of the resulting graphs. (a) A balanced graph with $N=12$ vertices and three symmetrically placed loops. (b) A balanced graph with $N=36$ vertices created by joining together three copies of the graph in (a) resulting in a graph with empty automorphism group.}
\label{fig:nonsymmetric}
\end{figure}
Besides the types of graphs discussed in section~\ref{sec:buildingblocks}, the building blocks may be used to construct balanced graphs with empty automorphism group, i.e. graphs without any graph symmetries. This can easily be achieved by combining the building blocks in such a way that the added graphs do no longer have trivial symmetry properties. A simple example obtained by connecting three cycle graphs with loops in an asymmetric way is shown in Fig.~\ref{fig:nonsymmetric}. In general, this shows that automorphism groups alone cannot be used to classify balanced graphs.
\section{Set of balanced states for $N<5$}
In this section, we show that for $N<5$, all balanced states may be constructed making use of the roots of unity. 

Consider the case $N=2$. We will show that the balanced states are given by the second roots of unity, i.e.~that $\mathbb{B}(2)=\left\{\vartheta_1=\alpha,\vartheta_2=\alpha+\pi,~\alpha\in\mathbb{S}^1 \right\}$.
\begin{proof}
We have
\begin{align*}
e^{i\vartheta_1}+e^{i\vartheta_2}=0 \Rightarrow e^{i(\vartheta_2-\vartheta_1)}=-1\\
\Rightarrow \vartheta_2-\vartheta_1=\pi,
\end{align*}
This means that in general, we have $\vartheta_1=\alpha_1\Rightarrow\vartheta_2=\alpha_1 + \pi$, which are the angles building the second roots of unity.
\end{proof}

Along the same lines, the set of balances states for $N=3$ is given by
\\$\mathbb{B}(3) =\left\{\vartheta_j=\frac{2\pi j}{3}+\alpha(t) ,\alpha(t)\in\mathbb{S}^1, j\in\{1,2,3\} \right\}$, i.e.~the third roots of unity.
\begin{proof}
Fixing $\vartheta_1=0$ without loss of generality, we get
\begin{align*}
&e^{i\vartheta_2}+e^{i\vartheta_3}=-1 \\
&\Rightarrow \sin(\vartheta_2)+\sin(\vartheta_3)=0\land \cos(\vartheta_2)+\cos(\vartheta_3)=-1\\
&\Rightarrow (\vartheta_2=-\vartheta_3~\lor~\vartheta_2=\pi+\vartheta_3) \land \cos(\vartheta_2)+\cos(\vartheta_3)=-1.
\end{align*}
Only the first solution, namely $\vartheta_2=-\vartheta_3$, is compatible with the second equation as $\vartheta_2=\pi+\vartheta_3$ yields $\cos(\vartheta_2)+\cos(\vartheta_3)=\cos(\pi+\vartheta_3)+\cos(\vartheta_3)=\cos(\vartheta_3)-\cos(\vartheta_3)=0\neq -1$. Thus we arrive at $2\cos(\vartheta_2)=-1$ which implies $\vartheta_2=\pm \frac{2\pi}{3}$, i.e.~$\vartheta_3=\mp \frac{2\pi}{3}$, which are the angles corresponding to the third roots of unity.
\end{proof}

Finally, a similar result can be proven for $N=4$. We will show that the balanced states are given by pairs of second roots of unity, i.e.~that \begin{align*}
\mathbb{B}(4)&=\Big\{\vartheta_j=\vartheta_k+\pi,\vartheta_l=\vartheta_m+\pi;j,k,l,m \in\{1,2,3,4\}\\
& \wedge \text{ indices non-equal}\Big\}.
\end{align*}

\begin{proof}
We have the following requirements on the angular variables
\begin{align*}
    e^{i\vartheta_1}+e^{i\vartheta_3}&=-(e^{i\vartheta_2}+e^{i\vartheta_4})~\tag{P1}\\
    e^{-i\vartheta_1}+e^{-i\vartheta_3}&=-(e^{-i\vartheta_2}+e^{-i\vartheta_4})~\tag{P2}.
\end{align*}
Now assume without loss of generality that $e^{i\vartheta_1}+e^{i\vartheta_3}\neq 0$ and $e^{i\vartheta_1}+e^{i\vartheta_4}\neq 0$. If one of the two equations held, we would have $\vartheta_3-\vartheta_1=\pi$ or $\vartheta_4-\vartheta_1=\pi$, respectively, using the proof for the second roots of unity and would thus be done. Now one can make use of the equality $e^{-i\vartheta_1}+e^{-i\vartheta_3}=(e^{i\vartheta_1}+e^{i\vartheta_3})/(e^{i\vartheta_1}\cdot e^{i\vartheta_3})$
to obtain the following result from equation (P2)
$$\frac{e^{i\vartheta_1}+e^{i\vartheta_3}}{e^{i\vartheta_1}\cdot e^{i\vartheta_3}}=\frac{e^{i\vartheta_2}+e^{i\vartheta_4}}{e^{i\vartheta_2}\cdot e^{i\vartheta_4}}.$$
Together with equation (P1), this implies that $$e^{i(\vartheta_2+\vartheta_4)}=e^{i(\vartheta_1+\vartheta_3)}$$ must hold since $e^{i\vartheta_1}+e^{i\vartheta_3}\neq 0$. Along the same lines, one may show that $$e^{i(\vartheta_1+\vartheta_4)}=e^{i(\vartheta_2+\vartheta_3)}$$ is true. Multiplying these two equations with each other yields the identities $e^{2i\vartheta_2}=e^{2i\vartheta_1}$ and $e^{2i\vartheta_3}=e^{2i\vartheta_4}$. This implies the two equalities $\vartheta_2=\vartheta_1 \lor\vartheta_2-\vartheta_1=\pi$ and
$\vartheta_4=\vartheta_3 \lor\vartheta_4-\vartheta_3 =\pi$. Only the latter possibilities $\vartheta_2-\vartheta_1=\pi$ and $\vartheta_4-\vartheta_3=\pi$ are compatible with equations (P1) and (P2) which completes the proof.
\end{proof}
\section{Summation over roots of unity}
\label{sec:rootssummation}
The following theorem states that $N^{\text{th}}$ roots of unity sum up to zero.
\begin{theorem}[Summation of roots of unity]
\label{theo:roots}
Let $N\in\mathbb{N}$ be a natural number and consider the vectors of $N^{\text{th}}$ roots of unity $\bm{r}_{N,k}=e^{ik\bm{\theta}_N}$, where $\bm{\theta}_N=(\theta_1,...,\theta_N)$ and $\theta_j=\operatorname{arg}(\rho_{j,N})=2\pi j/N\in\mathbb{S}^1$ are angular variables and $k\in \mathbb{N}^{<N}$ is a natural number. Then for any natural number $k\in\mathbb{N}$ the sum of the roots of unity vanishes
$$
\bm{1}^T\bm{r}_{N,k}=\sum_{j=1}^Ne^{ikj2\pi/N}=0.
$$
\end{theorem}
\begin{proof}
Multiplying both sides of the equation by $e^{ik2\pi/N}$ and subtracting the result from the original equation yields $(1-e^{ik2\pi/N})\bm{1}^T\bm{r}_{N,k}=\sum_{j=1}^Ne^{ikj2\pi/N}-\sum_{j=1}^Ne^{ik(j+1)2\pi/N}$. Now one can solve for $\bm{1}^T\bm{r}_{N,k}$, since $k<N$ and thus $1-e^{ik2\pi/N}\neq 0$ obtaining
\begin{align*}
\bm{1}^T\bm{r}_{N,k}&=\frac{\sum_{j=1}^Ne^{ikj2\pi/N}-\sum_{j=1}^Ne^{ik(j+1)2\pi/N}}{1-e^{ik2\pi/N}}\\
&=e^{ik2\pi/N}\frac{1-e^{ik2\pi}}{(1-e^{ik2\pi/N})}=0,
\end{align*}
where $e^{ik2\pi}=1,~\forall k \in\mathbb{Z}$ was used in the last step.
\end{proof}
\section{Balanced states in cycle graphs with loops}
\label{sec:balnonzeroloops}
\begin{figure}[tb!]
\centering
\includegraphics[width=.3\textwidth]{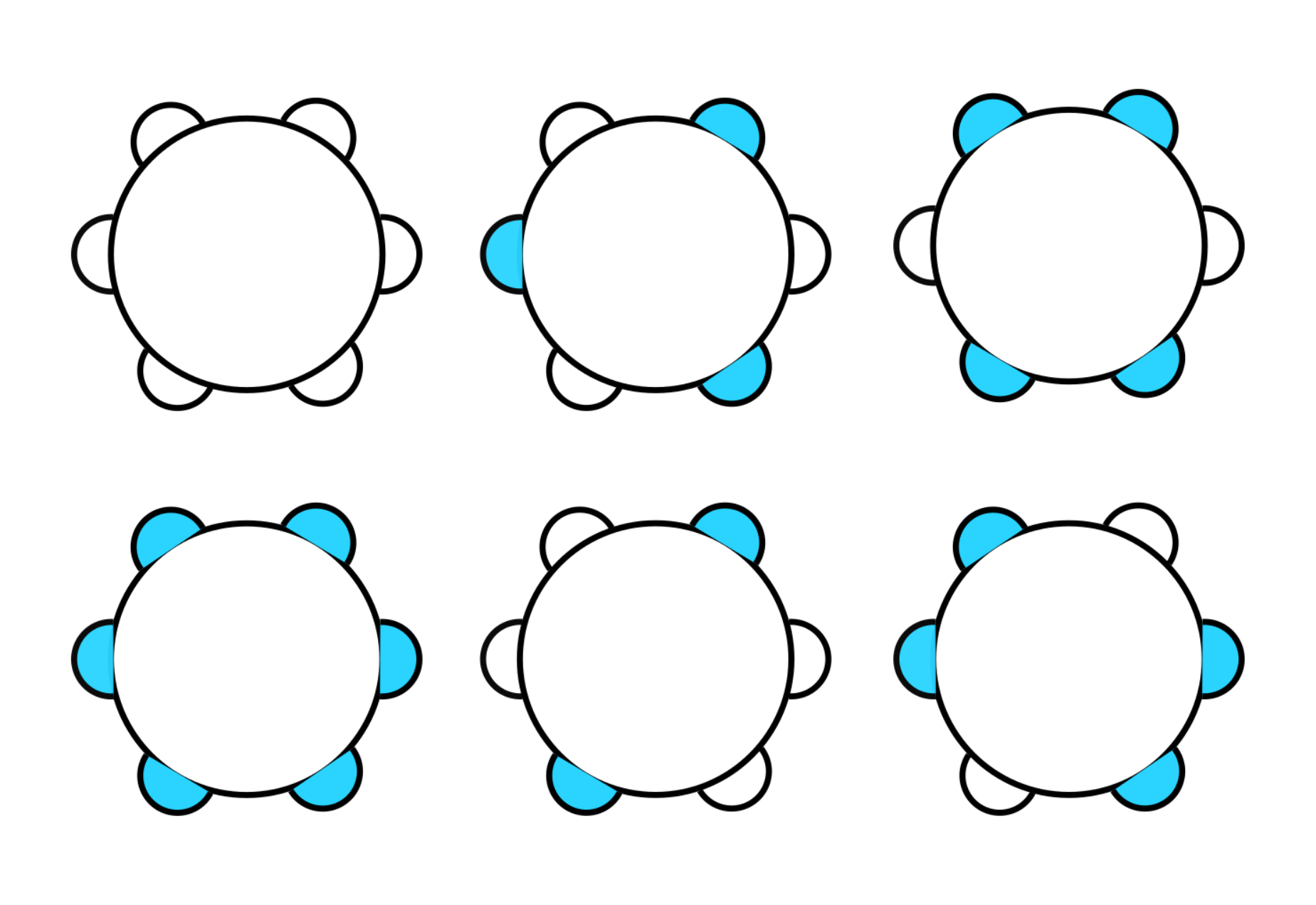}
\caption{Sketch of a cycle graph with $k=6$ loops. The winding numbers in the loops for balanced fixed points are not limited to be all-zero (top left) or all unity (bottom left, unity value indicated by colour), but also patterns where the loop winding numbers are prime divisors of $k$ may result in balanced fixed points (centre and right plots).}
\label{fig:SummationBalancing}
\end{figure}
In this section, we want to discuss the balanced states in cycle graphs with loops in more detail. Here, we will focus on the case where loops consist of $m=3$ vertices added to the graph for all loops such that the winding number of the loops characterizing stable fixed points reads $q_{l,i}\in\{-1,0,1\}$. To characterize these winding numbers, we will make use of results on the geometry of balanced states discussed in section~\ref{sec:phasebalancing}.
As an example, consider a cycle graph with $k=6$ loops. In this case, states with $\bm{q}_l\in\{(1,0,0,1,0,0)^T,(1,0,1,0,1,0)^T\}$ representing the loop number's prime factors can lead to balanced states as indicated in Fig.~\ref{fig:SummationBalancing}. To characterize these states on more mathematical grounds and express their relationship to the central loop's winding
number $q_1$, lets define the set of factors of the number of loops $k$ whose factors are not a divisor of some integer $q$ by $$R_{Q_l}(k,q)=\{p_1\in\mathbb{N}^{<k}\big|\exists p_2\in\mathbb{N}^{<k}:p_1\cdot p_2=k\wedge p_2\nmid q\},$$
where the expression $p_1\nmid q$ expresses the fact that $p_1$ is not a divisor of $q$.
Using this idea, one may write down the set of all loops' winding numbers leading to balanced states 
\begin{align*}
    Q_{l,2}^{>0}(q,k)&=\{\bm{q}_l\in\mathbb{D}_{l,2}^{>0}(\bm{q}_l)\big|\\ 
    &\bm{q}_l\in\operatorname{span}(\sigma^{p}((\underbrace{1,0,0,..}_{p_1 },...,\underbrace{1,0,0,..}_{p_1})^T),\\
    &p\in\{1,...,p_1-1\}),p_1\in R_{Q_l}(k,q)\},
\end{align*}
where $\sigma^{p}$ denotes again the cyclic shift applied $p$ times and the $^{>0}$ superscript indicates that this set is restricted to positive winding numbers. Note that this set only contains symmetric states if the integer $p_1$ is an integer divisor of the number of loops $k$ such that $p_1\big| k\in\mathbb{N}$, otherwise $R_{Q_l}$ is the empty set. 
The cyclic shift represents the graph's rotational symmetry. In general, the sign of the loops' possible winding numbers needs to coincide with the sign of the central cycle's winding number. 
Thus, one may classify the balanced fixed points with positive winding numbers by 
$$
Q_l^{>0}(k)=\{\bm{q}\in\mathbb{D}_{l}(\bm{q})\big|q_1\in Q^{>0}_{c}(k), \bm{q}_l\in Q^{>0}_{l,2}(q_1)\}.
$$
 Here, $Q_c=\{\bm{q}\in\mathbb{D}_c(\bm{q})\big|q_1=q_2,q_1\neq 0\}$ is the set defining the possible winding numbers in the inner cycle as defined in section~\ref{sec:results}. The set of negative winding numbers may be defined correspondingly and represents balanced states as well. Note that we now take the set of loop winding numbers in dependence of the winding number in the inner cycle $q_1$ $Q^{>0}_{l,2}(q_1)$ such that the main cycle's winding number induces the possible factors for the loops' winding numbers and thus their symmetry properties. The cases where $p_2|q_1$ in the set of factors $R_{Q_l}$ induces again a global symmetry similar to the one shown in Fig.~\ref{fig:SummationBalancing} b and the resulting states are not balanced. 
Note that this definition only characterizes balanced states for $m<5$ and needs to be redefined otherwise.

 In some cases, the combination of positive main winding number and negative loop winding numbers $q_1\in Q^{>0}_{c}(k), \bm{q}_l\in Q^{<0}_{l,2}(q_1)$ and vice versa leads to balanced states as well, that do, however, not always exist. In most cases, these states do exist for $q_1=1$. Nonetheless, these fixed points might possess phase differences larger than $\pi/2$, in particular for small graphs. Take for example the cycle graph $C_{40}$ to which $k=4$ loops have been added. Then the state
 characterized by $\bm{q}=(1,-1,0,-1,0)^T$ is stable and balanced, but has two phase differences larger than the  threshold $\varphi\approx\pi/2+0.141>\pi/2$. For this reason, these states may not be studied using the algorithm for finding stable, balanced states, but will be accounted for by the Monte Carlo method if they take a significant amount of the system's basin stability. However, most of the states with non-zero loop numbers occupy only a small share of basin stability in practice and may thus be neglected when
 calculating the graph's balancing ratio. 

Even though almost complete, this set still doesn't fully account for the balanced states. The stable fixed points that appear through zero-energy edges in higher order winding numbers of the corresponding cycle graph need to be included when summing over all balanced states. These states occur if the factor between the number of loops $k$ and the overall number of vertices in the new graph $N^*=N+k\cdot m$ is a multiple of the number of edges added in each loop $N^*/k=(m+1)\cdot p, p\in\mathbb{N}$. They are characterized by all equal, non-zero winding numbers in all loops $q_{l,j}\neq 0,~\forall j,$ and a non-zero winding number $q_1$ in the central cycle with the same sign $\operatorname{sign}(q_1)=\operatorname{sign}(q_{l,j})$ which reads $q_1=\operatorname{sign}\cdot(q_{l,j})(p-1)\cdot k$, see for example Fig.~\ref{fig:cyclebalanced} b. 
\end{document}